\documentclass[runningheads,11pt]{llncs}
\usepackage{amsmath,amssymb,amsfonts,mathrsfs%
}

\usepackage[pdftex]{graphicx}
\usepackage{algorithm}
\usepackage{algorithmic}
\usepackage{enumerate}
\usepackage{booktabs}
\usepackage{fullpage}

\usepackage[all]{xy}
\usepackage{color}
\usepackage{url}

\spnewtheorem{fact}[theorem]{Fact}{\bfseries}{\itshape}
\spnewtheorem{notation}[theorem]{Notation}{\bfseries}{\itshape}
\spnewtheorem{hypothesis}[theorem]{Assumption}{\bfseries}{\itshape}

\renewcommand{\leq}{\leqslant}

\renewcommand{\geq}{\geqslant}
\renewcommand{\ge}{\geqslant}

\newcommand{\eqdef}{\stackrel{\text{def}}{=}}

\newcommand{\F}{\ensuremath{\mathbb{F}}}

\newcommand{\prob}{\ensuremath{\textsf{Prob}}}
\newcommand{\code}[1]{\ensuremath{\mathscr{#1}}}
\newcommand{\Csec}{\code{C}_{\text{sec}}}
\newcommand{\Cpub}{\code{C}_{\text{pub}}}
\newcommand{\Crand}{\code{C}_{\text{rand}}}

\newcommand{\AC}{\code{A}}
\newcommand{\BC}{\code{B}}
\newcommand{\CC}{\code{C}}

\newcommand{\sqc}[1]{#1^2} 
\newcommand{\sqcp}[1]{\left(#1\right)^2}  
\newcommand{\stp}[2]{#1\cwp#2}   
\newcommand{\scp}[2]{#1\cdot #2} 
\newcommand{\cwp}{\star} 

\newcommand{\word}[1]{\vec{#1}}
\newcommand{\av}{\word{a}}
\newcommand{\bv}{\word{b}}
\newcommand{\alphav}{\word{\alpha}}
\newcommand{\betav}{\word{\beta}}

\newcommand{\cv}{\word{c}}

\newcommand{\ev}{\word{e}}

\newcommand{\mv}{\word{m}}
\newcommand{\pv}{\word{p}}
\newcommand{\uv}{\word{u}}
\newcommand{\vv}{\word{v}}

\newcommand{\xv}{\word{x}}
\newcommand{\yv}{\word{y}}

\newcommand{\mat}[1]{\ensuremath{\boldsymbol{#1}}}

\newcommand{\Dm}{\mat{D}}

\newcommand{\Gm}{\mat{G}}

\renewcommand{\Im}{\mat{I}}

\newcommand{\Pm}{\mat{P}}
\newcommand{\Pim}{\mat{\Pi}}

\newcommand{\Qm}{\mat{Q}}
\newcommand{\Rm}{\mat{R}}
\newcommand{\Sm}{\mat{S}}
\newcommand{\Tm}{\mat{T}}

\newcommand{\Ind}{{\mathcal I}}
\newcommand{\Jind}{{\mathcal J}}
\newcommand{\Gms}{\mat{G_{sec}}}
\newcommand{\Gmp}{\mat{G_{pub}}}

\newcommand{\Hmp}{\mat{H_{pub}}}

\newcommand{\fq}{\F_{q}}

\newcommand{\GRS}[3]{\text{\bf GRS}_{#1}\left(#2,#3\right)}

\newcommand{\sh}[2]{{\mathcal{S}}_{#1}\left(#2\right)}

\newcommand{\pu}[2]{{\mathcal{P}}_{#1}\left(#2\right)}

\title{\bf A Polynomial-Time Attack on the BBCRS Scheme} 

\author{Alain Couvreur \inst{1}
\and Ayoub Otmani \inst{2} 
\and Jean-Pierre Tillich \inst{3} 
\and  Val\'erie Gauthier--Uma\~na \inst{4}
}

\institute{
INRIA, {\sc Grace} Team \& LIX, CNRS UMR 7161,\\ All\'ee Honor\'e d'Estienne d'Orves, 91120
Palaiseau Cedex, France.\\
\email{alain.couvreur@lix.polytechnique.fr}
\and
University of Rouen,  LITIS, F-76821 Mont-Saint-Aignan, France.\\
\email{ayoub.otmani@univ-rouen.fr}
\and
INRIA, {\sc Secret} Team,
78153 Le Chesnay Cedex, France. \\
\email{jean-pierre.tillich@inria.fr}
\and
Faculty of Natural Sciences and Mathematics,\\ Department of Mathematics, Universidad del Rosario, Bogot\'a, Colombia.\\
\email{gauthier.valerie@urosario.edu.co}
}

\authorrunning{A.~Couvreur, A.~Otmani, J.--P.~Tillich, V.~Gauthier--Uma\~{n}a}

\begin{document}

\maketitle

\begin{abstract}
The BBCRS scheme is a variant of the McEliece public-key encryption scheme where the hiding phase 
is performed by taking the inverse of a matrix which is of the form  $\Tm + \Rm$ where
$\Tm$ is a sparse matrix with average row/column weight equal to a very small quantity $m$, usually $m < 2$,
and $\Rm$ is a matrix of small rank $z\ge 1$. The rationale of this new transformation is the reintroduction  
of families of codes, like generalized Reed-Solomon codes, that are famously known for representing 
insecure choices. We present a key-recovery attack when $z =1$ and $m$ is chosen between $1$ and $1+R+O(\frac{1}{\sqrt{n}})$
where $R$ denotes the code rate. This attack has complexity $O(n^6)$ and 
breaks all the parameters suggested in the literature.
\end{abstract}


\section*{Introduction}
{\bf Post-quantum cryptography.} All public key cryptographic primitives used in practice such as RSA, ElGamal scheme, DSA or ECDSA rely either on the difficulty of factoring or computing the discrete logarithm and would therefore be broken by Shor's algorithm \cite{Sho94a}
if a large enough quantum computer could be built. Moreover, even if a large enough  quantum computer might not be built in the next five years, it should be mentioned that tremendous
progress has been made for computing the discrete logarithm over finite fields of small characteristic with the quasi-polynomial time algorithm of \cite{BGJT14a}.  
This lack of diversity in public key cryptography has been identified as a major concern in the field of information security.
For all these reasons, it would be very desirable to be ready to replace these schemes by others
that would  rely on other hard problems. However only few other proposals have emerged which are essentially hash-based
 signature schemes, lattice-based, code-based and multivariate quadratic based schemes.
 They are  either based on the problem of solving multivariate equations over a finite field, the problem of finding a short vector in a lattice and the problem of decoding a linear code. Those problems are known for being NP-hard and are therefore  believed to be immune to the quantum computer threat.

{\bf The McEliece cryptosystem.}
Among those, one of the most promising scheme is the McEliece public key cryptosystem \cite{McEliece78}. It  is also one of the oldest public-key cryptosystem. 
It uses a family of codes for which there is a fast decoding algorithm (the binary Goppa code family here) which is used in the decryption process whereas an attacker 
has only a random generator matrix of the Goppa code  which reveals nothing about the algebraic structure of the Goppa code that is used in the decoding process. He has therefore 
to decode a generic linear code for which only exponential time decoding algorithms are known.  The main advantage of this system is to have very fast encryption and decryption functions. Depending on how the parameters are chosen for a fixed security level, this cryptosystem is about five times faster for encryption and about 10 to 100 times faster for decryption than RSA \cite{Sen08}. Furthermore, it has withstood many attacking attempts. After more than thirty five years now, it still belongs to the very few public key cryptosystems which remain unbroken. 

{\bf The use of Reed-Solomon codes in a McEliece scheme.}
Goppa codes are subfield subcodes of Generalized Reed-Solomon codes (GRS codes in short). This means that a Goppa code defined over $\fq$ is actually the set of codewords of 
a GRS code defined over an extension field $\F_{q^\mu}$ (we say that $\mu$ is the extension degree of the Goppa code) whose
coordinates all belong to the subfield $\fq$. Actually the fast decoding process of Goppa codes is the decoder of the underlying GRS code.
Roughly speaking, a Goppa code of length $n$ and dimension $n-2t\mu$ defined over $\fq$ can correct $t$ errors\footnote{but the dimension can be increased to
$n-t\mu$ in the binary case} and  is a subfield subcode of a GRS code that can also correct $t$ errors
which is of the same length $n$ but has a larger dimension  $n-2t$ and is defined over $\F_{q^\mu}$. In this sense, the underlying GRS code 
has a better error correction capacity than the Goppa code. This raises the issue of using GRS codes instead of Goppa codes
in the McEliece system. The better decoding capacity of GRS codes translates into smaller public key sizes for the McEliece scheme which is actually one of the main drawback of this scheme. 
This approach has been tried in Niederreiter's scheme (whose security is equivalent to the
McEliece scheme) but has encountered a dreadful fate when 
the Sidelnikov-Shestakov attack appeared \cite{SidelShesta92}. 

{\bf Baldi et al. approach for reviving GRS codes.} 
In their Journal of Cryptology article \cite{BBCRS11b}, Baldi \textit{et al.} have suggested a new way of using GRS codes in this context.
Instead of using directly such a code, they multiplied it by the inverse of the sum $\Tm + \Rm$ where $\Tm$ is a sparse matrix 
and $\Rm$ is a low rank matrix. By doing this, the attacker sees a code which is radically different from a GRS code but the
legitimate user can still use the underlying GRS decoder. This thwarts the Sidelnikov-Shestakov attack completely. 
However the decoding capacity of the resulting code is basically scaled down by a factor of 
$\frac{1}{m}$  where $m$ denotes the average weight of rows of the matrix $\Tm$.
It should be noted that the very same approach has also been tried for the Low-Density-Parity-Check code family, LDPC 
in short,
which is notoriously known for being insecure in a McEliece scheme \cite{MRS00,BaldiChiara07,BaldiChiara08}. 
In this case, they did not even use the low rank matrix and despite of this fact the resulting public code obtained by this multiplication is not 
an LDPC code anymore (it becomes a moderate-density-parity-check code) and it seems now that if the attacker wants to break this scheme
he has to be able to solve a generic decoding problem \cite{MTSB13}.
There are therefore good reasons to believe that this approach can be 
powerful for disguising the secret code structure.

{\bf An earlier attempt.} Baldi \textit{et al.} \cite{BBCRS11a} first used this approach with $\Tm$ being a permutation matrix. In this case $m=1$ and 
nothing is lost in term of decoding capacity compared to a GRS decoder. In other words, this allows to decrease the public key size as if we had
a GRS code in the McEliece cryptosystem. This first attempt got broken in \cite{CGGOT13,CGGOT:dcc14}. 
Roughly speaking the reason of this attack in this case can be traced back to two facts (i) it turns out that the 
resulting code is still close to the underlying GRS code: the intersection of the public code with the secret GRS code is of co-dimension one;
(ii) there is a very powerful way of distinguishing a GRS code \cite{CGGOT13} from a random code by computing the dimension of its square
which can be used to unravel the algebraic structure of the public code.
On the other hand, when the degree of sparseness of $\Tm$ is $>1$ the resulting code does not have a large intersection with a GRS code and there was some hope to obtain a secure scheme.

{\bf Our contribution: an attack which works in the regime $1 < m < 2$.}
In the present article we will show that despite the fact that the public code is far from being a GRS code, a similar trick that has already
been used to attack successfully in \cite{COT:EC14} some wild Goppa codes proposed in \cite{BLP10:wild} 
when the degree of extension is only $2$ can also be used in this context.
It consists in computing the
dimension of the square of shortenings of the public
code. Because of the hidden structure of the public code, the squares of
some of its shortenings have a smaller dimension than the squares of shortened
random codes of the same dimension.
This distinguisher is then used to unravel the structure of the matrix $\Tm$. This gives an attack of polynomial time complexity which 
can be used to break the examples given in \cite{BBCRS11b}. Several were broken in a few hours, and others in a few days. As an illustration, Example~1 given in \cite{BBCRS11b} with a claimed 90-bit security 
can be broken in 2.75 hours on  a computer equipped with Xeon 2.27GHz processor and 72 Gb of RAM.
This attack works up to values of $m$ of order $1+R+O(\frac{1}{\sqrt{n}})$, where $R$ is the rate of the public code.
The attack we present here can obviously be thwarted by taking values for $m$ greater than $2$,  
but in this case, since the price to pay is 
a decrease of the decoding capacity by a factor of more than $2$, we do not obtain better public key 
sizes than the ones we obtain by using Goppa codes, or more generally alternant
codes of extension degree $2$, 
provided we choose non wild Goppa codes in order to avoid the attack of \cite{COT:EC14}. 
The complexity of the present attack is similar to that of \cite{CGGOT:dcc14}, 
namely $O(n^6)$ where $n$ is the code length.
More precisely, this attack starts with two steps of respective complexity $O(n^3)$ and $O(n^5)$
and then applying the attack of \cite{CGGOT:dcc14} whose complexity is $O(n^6)$  operations in the base field.


\section {GRS Codes and the Square Code Construction}
  \label{sec:basics}

We recall in this section a few relevant results and definitions from coding theory and bring in the
fundamental notion
of square code construction.

\begin{definition}[Generalized Reed-Solomon code] \label{defGRS}
Let $k$ and $n$ be integers such that $1 \leqslant k < n \leqslant q$ where $q$ is a 
prime power.
The
code $\GRS{k}{\xv}{\yv}$ of dimension $k$ is associated to a pair
$(\xv,\yv) $ where $\xv$ is an $n$-tuple of distinct elements of
$\fq$  and
$\yv \in {(\fq^{\times})}^n$,
is 
defined as: 
$$
\GRS{k}{\xv}{\yv} \eqdef
\Big\{(y_1p(x_1),\dots{},y_np(x_n)) ~|~ p \in \fq[X], \deg p < k\Big\}.
$$
\end{definition}

The first work that suggested to use GRS codes in a public-key 
encryption
scheme 
was \cite{Nie86}. But Sidelnikov and Shestakov  
\cite{SidelShesta92} 
showed that for any GRS code it is 
possible  to recover in polynomial time a pair $(\xv,\yv)$ defining it,
which is all that is needed to decode efficiently such codes and is therefore enough to break
any McEliece type cryptosystem \cite{McEliece78} that uses GRS codes.

\begin{definition}[Componentwise products]
  Given two vectors $\av=(a_1, \dots, a_n)$ and $\bv=(b_1, \dots, b_n)
  \in \fq^n$, we denote by $\av \cwp \bv $ the componentwise product
 $$
 \av \cwp \bv \eqdef (a_1b_1,\dots{},a_n b_n).
 $$
\end{definition}
The star product $\av \cwp \bv$ should be distinguished from a more common operation,
namely the canonical inner product:
  $$
  \scp{\av}{\bv}\eqdef \sum_{i=1}^n a_i b_i.
  $$

\begin{definition}[Product of codes \& square code]
Let $\code{A}$ and $\code{B}$ be two codes of length $n$. The
\emph{star product code} denoted by $\code{A} \cwp \code{B}$ of $\code{A}$ and $\code{B}$ is \emph{the vector space
spanned by all products} $\av \cwp \bv$ where $\av$ and $\bv$ range over $\code{A}$ and $\code{B}$ respectively.
When $\code{B} = \code{A}$ then  $\code{A} \cwp \code{A}$ is called the \emph{square code} of $\code{A}$
and is rather denoted by $\sqc{\code{A}}$.
\end{definition}

\begin{proposition}\label{prop:dimprod}
Let $\code{A}$ be a code of length $n$, then
$$ \dim (\sqc{\code{A}}) \leq \min \left\{ n, \binom{\dim(\code{A})+1}{2}\right\}.$$
\end{proposition}

\begin{proposition}\label{prop:complexity}
  Let $\code{A} \subset \fq^n$ be a code of dimension $k$.
  The complexity of the computation of a basis of
  $\sqc{\code{A}}$ is $O(k^2n^2)$ operations in $\fq$.
\end{proposition}

See for instance \cite{CGGOT:dcc14}, for proofs of Propositions \ref{prop:dimprod} and \ref{prop:complexity}.

The importance of the square code construction becomes clear when we compare
the dimension of the square of 
structured codes like GRS codes with the
dimension of the square of a random code.
Roughly speaking, given a code of dimension $k$,  the dimension of its square
is linear in $k$ if it is a GRS code and quadratic if it is a random code
as explained in the two following  propositions.

\begin{proposition}\label{prop:square} 
$
\GRS{k}{\xv}{\yv}^2 = \GRS{2k-1}{\xv}{\yv\cwp\yv}.
$
\end{proposition}

\begin{proof}
See for instance \cite[Proposition 10]{MMP11a:DCC}.
\end{proof}

\begin{remark}
This property can also be used in the case $2k-1 >n$. To see this,
consider the dual of the Reed-Solomon code, which is itself a generalized Reed-Solomon code 
\cite[Theorem 4, p.304]{MacSloBook}.  
\end{remark}

\begin{theorem}
  Let $\code{A}$ be a random code of length $n$ and dimension $k$ such that
$n > {k+1 \choose 2}$. Then, for all integer $\ell<{k+1 \choose 2}$,
$$
\prob \left( \dim \sqc{\code{A}} \leq {k+1 \choose 2} - \ell \right)
= O \left( q^{-\ell} \cdot q^{-\left(n - {k+1 \choose 2}\right)} \right),\qquad (k\rightarrow +\infty).
$$
\end{theorem}

\begin{proof}
  See \cite{CCMZdraft}.
\end{proof}

\begin{remark}
A slightly weaker result was already obtained in the papers
\cite{FGOPT11a,FGOPT13}  (see also \cite{MP12a}).  
\end{remark}

For this reason, 
$\GRS{k}{\xv}{\yv}$ can be distinguished from 
 a random linear code of the same dimension by computing the dimension of the associated square 
code. In \cite{FGOPT11a,FGOPT13}, this phenomenon was already observed  for $q$-ary alternant codes (in particular Goppa codes) at very high rates whose duals are
distinguishable from random codes by the very same manner.
Subsequently, the very same phenomenon
lead to
attacks on GRS based cryptosystems \cite{CGGOT13,CGGOT:dcc14}, to
 a polynomial time attack on Wild Goppa codes over quadratic extensions \cite{COT:EC14} and to a polynomial time attack on algebraic geometry codes \cite{CMP_isit14}.

\medskip

Historically, the star product of codes  has been used for the first time by Wieschebrink to 
cryptanalyze a McEliece-like scheme \cite{BL05} based on subcodes of Reed-Solomon codes
 \cite{Wie10}.  The use of the star product here is nevertheless different from the way it is used in \cite{Wie10}. In Wieschebrink's paper,
 the star product is used to identify, given a certain low codimensional subcode  $\code{C}$ of a GRS code $\GRS{k}{\xv}{\yv}$, a possible pair $(\xv,\yv)$. This is achieved by computing $\sqc{\code{C}}$ which turns out to be $\GRS{k}{\xv}{\yv}^2 = \GRS{2k-1}{\xv}{\yv \cwp \yv}$
with a high probability. The Sidelnikov and Shestakov algorithm is then 
 used on  $\sqc{\code{C}}$ to recover a possible $(\xv,\yv\cwp\yv)$ pair to describe $\sqc{\code{C}}$ as a GRS 
 code, and hence, a pair $(\xv,\yv)$  is deduced for which $\code{C}
 \subset \GRS{k}{\xv}{\yv}$.

\section{Description of the Scheme}
   \label{sec:schemeit}

The BBCRS  public-key encryption scheme given in \cite{BBCRS11b} can be summarized as follows:

\begin{description}
	\item \textbf{Secret key.} 
          \begin{itemize}
          \item $\Gms$ is a generator matrix of a GRS code of length $n$ and dimension $k$ over $\fq$.
          \item $ \Qm  \eqdef \Tm + \Rm $ where
            $\Tm$ is an $n \times n$ non-singular sparse matrix with elements in $\fq$ and average row weight $m \ll n$. Note that $m$ is not necessarily an integer. For example $m=1.4$ means that $40\%$ of the rows of $\Tm$ have weight equal to 2 and the other $60\%$ have weight equal to 1.
          \item $\Rm$ is a rank-$z$ matrix over $\fq$ such that $\Qm$
            is invertible. In other words there exist 
             $\alphav \eqdef (\alpha_1, \dots{}, \alpha_n) $ and $\betav \eqdef (\beta_1, \dots{}, \beta_n)$ 
 such that  $\Rm \eqdef \alphav^T \betav$ and $ \alpha_i$ and $\beta_i$ are $z\times 1$ full rank matrices defined over $\fq$ for all $ i \in \{1, \dots, n \}$ and $z\leq n$.
          \item $\Sm$ is a $k \times k$ random invertible  matrix over $\fq$.
          \end{itemize}

        \item \textbf{Public key.} 
        \begin{equation}\label{eq:Gpub}
        \displaystyle \Gmp \eqdef \Sm^{-1} \Gms \Qm^{-1}.
        \end{equation}
          
	\item \textbf{Encryption.} The ciphertext $\cv \in \fq^n$ of a plaintext
          $\mv \in \fq^k$ is obtained by drawing at random $\ev$
          in $\fq^n$ of weight less than or equal  to $\frac{n-k}{2m}$ (recall that $m$ denotes the density of the matrix $\Tm$) and computing
          $\displaystyle \cv \eqdef \mv \Gmp  +  \ev$. 
          
	\item \textbf{Decryption.} It consists in performing the three
          following steps:
	\begin{enumerate}
	\item Guessing the value of  $\ev  \Rm$.
	\item Calculating $\cv' \eqdef \cv \Qm - \ev \Rm= \mv \Sm^{-1}\Gms + \ev  \Qm - \ev \Rm =  \mv \Sm^{-1}\Gms + \ev  \Tm $
	and using the decoding algorithm of the GRS code to recover
	$\mv \Sm^{-1}$ from the knowledge of $\cv'$.
	\item Multiplying the result of the decoding by $\Sm$ to recover $\mv$.
	\end{enumerate}
\end{description}

\begin{remark}
  In \cite{BBCRS11b}, the authors suggest to take $m = 1+ \frac{n-k-3}{n} \approx 2-R$ for the density of $\Tm$.
\end{remark}

\paragraph{\bf Further details on the construction of the matrix $\Tm$.}

We deal with the case $m \leq 2$. 
According to \cite{BBCRS11b} the matrix $\Tm$ is constructed\footnote{Actually, the authors propose three constructions for $\Tm$
and express a clear preference for the one described in the present
article.}  as follows.
\begin{enumerate}
\item Choose a permutation matrix $\Pm$. Replace each $1$ by a random 
element of $\fq^{\times}$.
\item Set $t\eqdef \lfloor \frac{n-k}{2} \rfloor$, $\delta_t \eqdef t - \lfloor \frac{t}{m} \rfloor$ and $\ell \eqdef \lfloor (m-1)n \rfloor$. Choose a random set $\mathcal{C}$ of $\delta_t$
columns and a random set $\Jind_2$ of $\ell$ rows of $\Pm$.
\item For all $i \in \Jind_2$, we denote by $\pi(i)$ the integer
such that $\Pm_{i, \pi(i)} \neq 0$. For each $i \in \Jind_2$, choose a random
element $j \in \mathcal{C}\setminus \pi(i)$ and add a random element of $\fq^{\times}$ at position $(i,j)$.
\end{enumerate}

We also tested another construction allowing to have row and column
weight upper bounded by $2$.
The sparse matrix $\Tm$ is constructed as $\Tm = \Tm_1 + \Tm_2$
where:
\begin{itemize}
\item $\Tm_1$ is of the form $\Tm_1 = \Dm_1 \Pm_1$, where $\Dm_1$ is diagonal
invertible and $\Pm_1$ is a permutation matrix;
\item $\Tm_2 = \Dm_2 \Pm_2$, where $\Dm_2$ is diagonal with $(m-1)n$ nonzero
diagonal coefficients and $\Pm_2$ is a permutation matrix; 
\item The matrices do not overlap, that is, there is no pair $(i,j)$
with $1\leq i, j \leq n$ such that both $(\Tm_1)_{ij}$ and $(\Tm_2)_{ij}$
are nonzero.
\end{itemize}

Our attack works for both choices of the matrix $\Tm$. The
experimental results in Sec. \ref{sec:exp_res} rely on the
first construction for $\Tm$.
 
\subsection{Previous attacks and discussion on the parameters}
\label{ss:total_recall}

The BBCRS scheme has been subject to an attack \cite{CGGOT:dcc14} in the case
$m=1$, \textit{i.e.} the matrix $\Tm$ is a permutation matrix and $z=1$, \textit{i.e.} the matrix
$\Rm$ has rank $1$. The attack presented here holds for $m < 1+R + O(\frac{1}{\sqrt{n}})$ and $z=1$. 
The relevance of choosing
higher $m$ or $z$ is discussed in Section \ref{sec:conclusion}.

The attack of the present article uses in its last step  the attack \cite{CGGOT:dcc14} on the original system \cite{BBCRS11a}.

\subsection{Notation}
It will be convenient to bring the following notation.
\begin{itemize}
\item $\Cpub$ is  the code with generator matrix $\Gmp$;

\item $\Csec$ is the GRS code with generator matrix $\Gms$, we assume that it is specified by its dual (which is itself a GRS code) as
$\Csec^\perp = \GRS{n-k}{\xv}{\yv}$;

\item $\Jind_1$ is the set of positions which correspond to rows of $\Tm$ of Hamming weight $1$. The elements of
$\Jind_1$ are called the {\em positions of degree $1$}. For any row $i \in \Jind_1$ of $\Tm$, we define $j(i)$ as the unique column of $\Tm$ for which $T_{ij(i)} \neq 0$;

\item $\Jind_2$ is the set of positions which correspond to rows of $\Tm$ of Hamming weight $2$.
The positions in $\Jind_2$ are called the {\em positions of degree $2$}.
When $i$ belongs to $\Jind_2$, let $j_1$ and $j_2$ be
the columns of $\Tm$ for which we have $T_{ij_1} \neq 0$ and $T_{ij_2} \neq 0$.
We define similarly $j(i)$ as the set $\{j_1,j_2\}$ in this case. 
\end{itemize}

\subsection{Structure of the public code}

The following result explains how $\Cpub$ and $\Csec$ and their duals are related.
\begin{lemma}
\label{lem:basics}
\begin{eqnarray}
\Cpub & = &\Csec (\Tm+\Rm)^{-1} \label{eq:Cpub}\\
\Cpub^\perp & = & \Csec^\perp (\Tm + \Rm)^T \label{eq:Cpubperp}.
\end{eqnarray}

\end{lemma}

\begin{proof}
The first equality follows immediately from \eqref{eq:Gpub}, whereas 
the second one was is observed in \cite[p.6, Equation (8)]{BBCRS11b} 
where a parity-check matrix for the public code $\Cpub$ is expressed in
terms of a parity-check matrix of the secret code. This can be proved as
 follows.
For all $c\in \Csec$, $c'\in \Csec^{\bot}$, 
$$
\scp{(c(\Tm +\Rm)^{-1})}{(c'(\Tm+\Rm)^T)} = \scp{(c(\Tm +\Rm)^{-1}(\Tm+\Rm))}{c'}
= \scp{c}{c'}=0.
$$
Moreover, since $\Qm = \Tm + \Rm$ is invertible, 
we get $\dim \Csec^{\perp}(\Tm+\Rm)^T + \dim \Csec(\Tm +\Rm)^{-1} = n$,
hence the codes are dual to each other.
\end{proof}

\section{The fundamental tool: shortening and puncturing the dual of the public code}

The lemmas stated in the present subsection are proved in Appendix \ref{sec:A}.

Puncturing and shortening will play a fundamental role in the attack. Recall that for a given code $\CC \subset \fq^n$ and a subset $\Ind$ of code positions
 the \emph{punctured} code $\pu{\Ind}{\CC}$ and \emph{shortened} code  $\sh{\Ind}{\CC}$
are defined as:
$$
\begin{aligned}
\pu{\Ind}{\CC} &\eqdef \left \{{(c_i)}_{i \notin \Ind}~|~ \cv \in \CC \right \}; \\
\sh{\Ind}{\CC} &\eqdef \left \{ {(c_i)}_{i \notin \Ind}~|~\exists \cv={(c_i)}_i \in \CC \text{ such that } 
\forall i \in \Ind,\; c_i = 0 \right \}.  
\end{aligned}
$$

Given a subset $\Ind$ of the set of coordinates of a vector $\uv$, we denote by
$\pu{\Ind}{\uv}$ the vector $\uv$ \emph{punctured} at $\Ind$, that is to say, \emph{indexes that are in $\Ind$ are removed}.

First let us recall the influence of these operations on GRS codes.
  \begin{lemma}\label{lem:punc_short_GRS}
    Let $\xv, \yv$ be two $n$--tuples of element sof $\fq$
    such that $\xv$ has pairwise distinct entries and $\yv$
    has only nonzero entries. Let $k<n$ and $\Ind \subseteq \{1, \ldots, n\}$.
    Then
    \begin{align}
      \label{eq:punctGRS} 
      \pu{\Ind}{\GRS{k}{\xv}{\yv}} &= \GRS{k}{\pu{\Ind}{\xv}}{\pu{\Ind}{\yv}}\\
      \label{eq:shortGRS}
      \sh{\Ind}{\GRS{k}{\xv}{\yv}} &= \GRS{k-|\Ind|}{\pu{\Ind}{\xv}}{\yv_{\Ind}},
    \end{align}
    for some $\yv_{\Ind} \in \fq^{n-|\Ind|}$ depends only on $\yv$ and $\Ind$.
  \end{lemma}

Next, with these notions at hand,  it follows that the dual of the public code punctured in $\Jind_2$ is very close to a GRS code. 
We will also need to understand the structure of versions of this code which are shortened in positions belonging to 
$\Jind_1$ and then punctured in $\Jind_2$. It turns out that these codes too are close to 
GRS codes.
First of all, puncturing $\Cpub^\perp$ in the positions belonging to $\Jind_2$ gives ``almost'' a GRS code, as shown by:
\begin{lemma}
\label{lem:puncture_1}
Let $\uv=(u_i)_{i \in \Jind_1}$ and $\vv=(v_i)_{i \in \Jind_1}$ be vectors in $\fq^{n-|\Jind_2|}$ defined by
\begin{eqnarray*}
u_i & = & x_{j(i)}\\
v_i & = & T_{ij(i)} y_{j(i)}.
\end{eqnarray*}
Let $\code{D} \eqdef \Csec^\perp \Tm ^T$, then 
\begin{equation}
\pu{\Jind_2}{\code{D}} \subseteq \GRS{n-k}{\uv}{\vv}.
\end{equation}
\end{lemma}

\begin{lemma}\label{lem:puncture_2}
Let $\lambda$ and $\mu$ be vectors of $\fq^n$ such that
$\Rm^T = \lambda^T \mu$ and let $\Csec^\perp(\lambda) \eqdef  \Csec^\perp \cap {<\lambda>}^\perp$, $\Cpub^\perp(\lambda) \eqdef \Csec^\perp(\lambda)(\Tm^T + \Rm^T)$.
Then, 
\begin{equation}
\label{eq:punctured}
\pu{\Jind_2}{\Cpub^\perp(\lambda)} \subseteq \GRS{n-k}{\uv}{\vv},
\end{equation}
Moreover if $\Jind_1$ contains an information set\footnote{
In coding theory, an {\em information set} of a code $\CC$
of dimension $k$
is a set of $k$ positions $\Ind$ such that the knowledge of 
a codeword $\cv\in \CC$ on the positions in $\Ind$ determines
entirely the codeword. Equivalently, if $\Gm$ denotes a 
$k\times n$ generator matrix of the code, then the $k\times k$
submatrix of $\Gm$ given by extracting the columns indexed by $\Ind$
is invertible.}
of $\Csec^\perp \Tm^T$ and $\Tm^T$ is invertible, then 
there exist $\av$ and $\bv$ in $\fq^{n-|\Jind_2|}$ such that for any $\cv$ in $\pu{\Jind_2}{\Cpub^\perp}$, there exists a vector $\pv$ in $\GRS{n-k}{\uv}{\vv}$ for which
\begin{equation}
\label{eq:further}
\cv = \pv +(\scp{\pv}{\bv})\av.
\end{equation}
In particular, $\pu{\Jind_2}{\Cpub^{\perp}} \subseteq \GRS{n-k}{\uv}{\vv} +
< \av >$.
\end{lemma}

 If we puncture with respect to $\Jind_2$ shortened versions of $\Cpub^\perp$ in positions belonging to $\Jind_1$, then we observe
 a similar phenomenon, namely 
 \begin{lemma}
 \label{lem:inclusion_punctured_shortened}
Let $\Ind_1$ be a subset of code positions which is a subset of $\Jind_1$. Let $s \eqdef |\Ind_1|$ and assume that 
$s \leq n-k$. Then there exist vectors $\av,\uv,\vv$ in  $\fq^{n-s-|\Jind_2|}$
such that: 
\begin{equation}
\label{eq:punctured_shortened}
\pu{\Jind_2}{ \sh{\Ind_1}{\Cpub^\perp}} \subseteq \code{E} + <\av>
\end{equation}
and $\code{E}$ is a subcode of $\GRS{n-k-s}{\uv}{\vv}$.
\end{lemma}

\section{Key-Recovery Attack}

\subsection{Outline}

Our key-recovery attack 
starts with a parity-check matrix $\Hmp$ of the (public) code $\Cpub$. The main goal is to recover matrices $\Tm$ and $\Rm$,
where $\Hmp {(\Tm^T + \Rm^T)}^{-1}$ is a parity check matrix
of a GRS code, $\Tm$ is a low density square matrix and $\Rm$
a rank $1$ matrix. Recall that in our terminology, rows of $\Tm$ belonging to $\Jind_1$ are positions of degree 1, and those in $\Jind_2$ are positions of degree 2. It implies, thanks to \eqref{eq:Cpubperp}, that some columns of $\Hmp$ belong to $\Jind_1$ and the others are in $\Jind_2$. 

Our attack is composed of three mains steps having the following objectives:
\begin{enumerate}
	\item Detecting columns of $\Hmp$ that belong to $\Jind_2$, and then deducing those of $\Jind_1$.
	\item Transforming columns of  $\Jind_2$ into degree 1 columns by linear combinations with columns of $\Jind_1$.
	
	\item At this stage, the public code has been  transformed into another code $\CC$ such that there exists a secret GRS code $\Csec'$ 
and a matrix $\Pi+ \Rm'$ where $\Pi$ is a permutation matrix and $\Rm'$
is rank-1 matrix such that:
	\begin{equation}
	\CC = \Csec' (\Pi + \Rm').
	\end{equation}
	The third step consists then in applying the attack developed in \cite{CGGOT:dcc14} which is purposely devised  to recover a pair $(\Pi, \Rm')$ from $\CC$ as outlined in Section~\ref{ss:total_recall}.
\end{enumerate} 

The purpose of the next sections is to describe more precisely the first two steps of the attack. Finally, the algorithms used in our implementation are postponed in Appendix~\ref{sec:algo}.
\subsection{A distinguisher of the public code}

The attack uses in a crucial way a distinguisher which discriminates the public code from a random code of the same dimension.
It  is based on square code considerations.
The point is the following: if we shorten the dual $\Cpub^\perp$ of the public code in a large enough set of positions $\Ind$, then 
the square code $\sqcp{\sh{\Ind}{\Cpub^\perp}}$ has dimension strictly smaller than that of
$\sqcp{\sh{\Ind}{\Crand^\perp}}$ where $\Crand$ is a random code of the same dimension as $\Cpub$.
The code $\sqcp{\sh{\Ind}{\Crand^\perp}}$
has dimension which is typically $\min\left\{n-|\Ind|,\binom{k_{\Ind}+1}{2} \right\}$ where $k_{\Ind}$ stands for the dimension of 
$\sh{\Ind}{\Crand^\perp}$. In general, $k_{\Ind}$ is equal to $n-k- |\Ind|$ since $\dim\Crand^\perp = \dim \Cpub^\perp =  n-k$  whereas we generally have: 
\begin{equation}
\label{eq:upper_bound}
 \dim \sqcp{\sh{\Ind}{\Cpub^\perp}}\leq 3(n-k) + |\Jind_2| - 3|\Ind| -1.
\end{equation}
In other words, when $3(n-k) + |\Jind_2| - 3|\Ind| -1 < \min\left\{n-|\Ind|,\binom{k_{\Ind}+1}{2} \right\}$ we expect to 
distinguish $\Cpub$ from a random code of the same dimension.
We write here ``generally'' because there are some exceptional cases where such an inequality does not hold. However in the
case when $\Ind \subset \Jind_1$, this inequality always holds.

\begin{proposition}
\label{prop:upper_bound}
Let $\Ind \subseteq \Jind_1$, then $
 \dim \sqcp{\sh{\Ind}{\Cpub^\perp}} \leq 3(n-k)-3|\Ind|  - 1  + |\Jind_2|.$
\end{proposition}

This proposition is proved in Appendix \ref{sec:upper_bound}.

\begin{remark}\label{rem:upper_bound}
It turns out that a similar inequality also generally holds when $\Ind$ contains degree $2$ positions.
However in this case, the situation is more complicated and it might happen in rare cases that 
this upper-bound is not met but, roughly speaking, when it happens, 
the actual result remains {\em close} to this upper bound. An explanation of what happens in this case is given in Appendix \ref{sec:heuristic}.
Experimentally, we observed that (\ref{eq:upper_bound}) was satisfied 
even when $\Ind$ contained positions of $\Jind_2$.
\end{remark}

\begin{remark}\label{rem:distnoshort}
  The use of shortening is important since in general the (dual) public code
  itself is non distinguishable because its square equals the whole ambient
  space.
  However, for a part of the parameters proposed in \cite{BBCRS11b},
  the dual public code is distinguishable from a random code without shortening.  See \S \ref{sec:exp_res} for further details.
\end{remark}

\subsection{Description of the attack}
\subsubsection*{First step -- Distinguishing between positions in $\Jind_1$ and $\Jind_2$}\label{ss:distinguishing}

Roughly speaking the attack builds upon an algorithm which allows to distinguish between a position of degree $1$ and a position of degree $2$.
It turns out now that once we are able to distinguish the public code from a random one by shortening it in a set of positions $\Ind$ such that:
\begin{equation}\label{eq:dist_rule}
 \dim \sqcp{\sh{\Ind}{\Cpub^\perp}} <
\min\left\{n-|\Ind|,\binom{n-k-|\Ind|+1}{2} \right\},
\end{equation}
we can puncture 
$\sh{\Ind}{\Cpub^\perp}$ in a position $i$ that does not belong to $\Ind$ and this allows to 
distinguish degree $1$ positions from degree $2$ positions.
The dimension of the square code of this punctured code will differ drastically when $i$ is a degree $1$ position (or a certain type of degree $2$ position)
or a ``usual''  degree $2$ position.
When $i$ is a degree $1$ position it turns out that
\begin{equation}
\label{eq:pas_clair}
\dim \sqcp{\sh{\Ind}{\Cpub^\perp}} =  \dim \sqcp{\pu{i}{\sh{\Ind}{\Cpub^\perp}}},
\end{equation}
whereas for ``usual'' degree $2$ positions we observe that 
\begin{equation}
\label{eq:normal} 
\dim \sqcp{\sh{\Ind}{\Cpub^\perp}} =  \dim \sqcp{\pu{i}{\sh{\Ind}{\Cpub^\perp}}}+ 1.
\end{equation}
Sometimes (in the ``non usual'' cases), we can have positions of degree $2$ for which 
$$\dim \sqcp{\sh{\Ind}{\Cpub^\perp}} =  \dim \sqcp{\pu{i}{\sh{\Ind}{\Cpub^\perp}}}$$
as for degree $1$ positions. This happens for instance if shortening in $\Ind$ ``induces'' a degree $1$ position in $i$.
This arises mostly when the position $i$ of degree $2$ is such that 
$j(i)=\{j_1,j_2\}$ where either $j_1=j(i')$ or $j_2=j(i')$ for a position $i'$ of degree $1$ that belongs to $\Ind$.
Further details on these phenomena are given in Appendix \ref{sec:heuristic}. This phenomenon really depends on the choice
of $\Ind$. However, by choosing several random subsets $\Ind$ we quickly find a shortening set $\Ind$ for which the 
degree $2$ position we want to test behaves as predicted in \eqref{eq:normal}.
This yields an algorithm to decide whether a given position has degree 2. See Algorithm \ref{al:detection2} in Appendix \ref{sec:algo}.

Moreover, we explain below, how to use the above observations to compute
the whole set of positions of degree $2$.


\smallbreak

\noindent {\bf Procedure to compute $\Jind_2$}
\begin{itemize}
\item Choose a set of random subsets $\Ind_1, \ldots, \Ind_s$ (in our 
experimentations we always chose $s \approx 20$) whose cardinals
satisfy \eqref{eq:dist_rule}.
\item For $i = 1, \ldots, s$ compute $\sqc{\sh{\Ind_i}{\Cpub^{\perp}}}$
and call $\Jind_2(i)$ this set of positions satisfying
$$
\dim \sqc{\sh{\Ind_i}{\Cpub^{\perp}}} \neq \dim
\pu{j}{\sqc{\sh{\Ind_i}{\Cpub^{\bot}}}}.
$$
\item Set $\Jind_2 = \Jind_2(1) \cup \cdots \cup \Jind_2(s)$.
\end{itemize}

The above described procedure
of degree $2$ is summarized by Algorithm \ref{al:compute_degree2}
(see Appendix \ref{sec:algo}).

\subsubsection*{Second step -- Transforming degree 2 positions into degree 1 ones}\label{ss:finishing}
This step reposes on the following statements
proved in Appendix \ref{sec:plus}.

\begin{proposition}
\label{pr:plus}
Let $i_1\in \Jind_1$ and $i_2 \in \Jind_2$ be a position associated to $i_1$.
Let $\Dm(\alpha,i_1,i_2)$ be an $n \times n$ matrix which is the identity matrix with an additional entry in column $i_2$ and row $i_1$
that is equal to $\alpha$.
Define $\code{C} \eqdef \Cpub^\perp \Dm_{\alpha,i_1,i_2}$. 
If $\alpha = -\frac{T_{i_2j_1}}{T_{i_1j_1}} $, then there exists  $\Rm'$ of rank at most one
such that  
\begin{equation}
\label{eq:two_one}
\code{C} = \Csec^\perp ({{\Tm}'}^T + {{\Rm}'}^T)
\end{equation}
where $\Tm'$ differs from $\Tm$ only in row $i_2$ and column $j_1$, the corresponding entry being now equal to 
$0$.
\end{proposition}

This proposition is exploited as follows.
\begin{itemize}
\item  We first compute for a degree $1$
position $i_1$ the set of degree $2$ positions $i_2$  such that
$j(i_1) \in j(i_2)$.
These positions $i_2$ can be detected by checking if 
$i_2$ has now become a degree $1$ position for $\sh{\{i_1\}}{\Cpub^\perp}$ (this is the case if and only if $j(i_1) \in j(i_2)$).
See Algorithm \ref{al:Associated} in Appendix \ref{sec:algo}.
\item Once such a pair  $(i_1, i_2)$ has been found we try all possible values for $\alpha \in \fq^\times$ until we obtain a code $\code{C}$ for which the corresponding $\Tm'$  contains a row of index $i_2$ which 
is now of Hamming weight $1$. That is to say: $i_2$ became a position of degree
$1$ for $\CC$. This can be easily checked by using the previous technique to 
distinguish between a position of degree $1$ or $2$.
See Algorithm \ref{al:Eliminate}, in Appendix \ref{sec:algo}.
\item In other words, when we are successful, we obtain a new code $\code{C}$ for which there is one more row of weight $1$. We iterate this process
by replacing $\Cpub^\perp$ by $\code{C}$ and $\Jind_1$ by $\Jind_1 \cup \{i_2\}$ until we do not find such pairs $(i_1, i_2)$.
For the values of $m$ chosen in \cite{BBCRS11b} and with rows of $\Tm$ which were all of 
weight $1$ or $2$ we ended up with $\Tm'$ which was a permutation matrix and a code $\code{C}$ which was linked to the secret code by 
$$
\code{C} = \Csec^\perp (\Pi + \Rm')
$$
where $\Pi$ is a permutation matrix and $\Rm'$ a matrix of rank at most $1$. 
To finish the attack, we just apply the attack described in \cite[Sec.4 ]{CGGOT:dcc14} to recover $\Csec$.
\end{itemize}

Algorithm \ref{al:complete_attack} given in Appendix \ref{sec:algo} describes
the complete attack.


\subsection*{Case of remaining degree-$2$ positions}

It could happen that the previoulsy decribed method is unsufficient to
transform every degree $2$ position into a degree $1$.
It could for instance
happen if there is a position $i$ of degree $2$ such that
for all position $i'$ of degree $1$, $j(i') \notin j(i)$.
In such a situation, no position of degree $1$ can be used to eliminate this position of degree $2$.

This problem can be addressed as soon as the set of positions of degree
$1$ contains an information set of the code.
We describe the strategy to conclude the attack in such a situation.

Let $\CC$ be the code obtained after performing the two steps of
the attack and assume that there remains
as nonempty set $\Jind_2$ of positions of degree $2$,
which are known (since they have been identified during the first step of
the attack).
Here is the strategy
\begin{enumerate}
\item Puncture $\CC$ at $\Jind_2$. The punctured code
is of the form 
\begin{equation}\label{eq:puJind2}
\CC' (\Im + \Rm')  
\end{equation}
where $\CC'$ is a GRS code,
 $\Im$ is the identity matrix and $\Rm'$ a rank $1$ matrix.
\item Perform the attack of \cite{CGGOT:dcc14} on $\pu{\Jind_2}{\CC}$.
We get the knowledge of a support $\xv'$ a multiplier $\yv'$
and a rank $1$ matrix $\Rm'$ such that
$$
\CC' = \GRS{k}{\xv'}{\yv'} (\Im + \Rm').
$$
Moreover, we are able to identify the polynomials $P_1, \ldots, P_k$
yielding the rows of the public matrix $\Gmp$. 
\item For all $x\in \fq$ which is not in the support $\xv'$ of $\CC'$, 
compute the column
$$
\left(
  \begin{array}{c}
    P_1(x) \\ P_2(x) \\ \vdots \\ P_k(x)
  \end{array}
\right)
$$
and join it to the matrix $\Gmp$. By this manner we get new positions of degree
$1$ which can be used to eliminate the remaining positions of degree $2$.
\end{enumerate}

\begin{remark}
In our experiments, this situation never happened: we
have always eliminated all the degree $2$ positions using Proposition \ref{pr:plus}.
\end{remark}





\section{Limits and Complexity of the Attack}
\subsection{Choosing appropriately the cardinality of $\Ind$}

By definition of the density $m$, the sets $\Jind_1$ and
$\Jind_2$ have respective cardinalities $(2-m)n$ and $(m-1)n$.
In what follows, we denote by $R$ the rate of the public code namely $ R= k/n$.
Let us recall that the attack shortens the dual of a public code which is of dimension $n -k$. 
The cardinality of $\Ind$ is denoted by $a$.
We list the constraints we need to satisfy for the success of the attack.
\begin{enumerate}
\item The shortened code should be reduced to the zero space, which implies that $a < n - k$.

\item The code punctured at $\Jind_2$ must contain an information set, that is to say:
  \begin{equation}
    \label{eq:constraint1}
    n-k \leq |\Jind_1|.
  \end{equation}
 It is clear that \eqref{eq:constraint1} is equivalent to $m \leq 1 + R$.
 
\item The computed square code in Proposition \ref{prop:upper_bound}
should also be different from the full space which implies:
\begin{equation}
  \label{eq:constraint2}
 3(n-k-a)+|\Jind_2| - 1 < n-a 
\end{equation}
One can easily check that \eqref{eq:constraint2} is equivalent to:
\begin{equation}\label{eq:amin}
a \geq \frac{1}{2} \Big( (1+m) n - 3 k \Big).
\end{equation}

\item Finally, to have good chances that the dimension of the square code
reaches the upper bound given by Proposition \ref{prop:upper_bound}, we also need:
\begin{equation}
  \label{eq:constraint3}
3(n-k-a)+|\Jind_2|-1 <   {n-k-a+1 \choose 2} 
\end{equation}
which is equivalent to the inequality:
\begin{equation}\label{eq:constraint3_2}
a^2 + \Big(5 - 2(n-k)\Big) a + (n-k)^2 - 5(n-k) +2(1-m)n \ge 0
\end{equation}
Considering \eqref{eq:constraint3_2} as an inequality involving a  degree-2 polynomial in $a$, we can 
check that its discriminant is equal to $\Delta \eqdef 8(m-1)n +25$, so that its roots are $a_0$ and $a_1$ where:
\begin{equation} \label{a:limit}
a_0 \eqdef n-k -\frac{5}{2} - \frac{1}{2}\sqrt{\Delta}
\text{~~ and ~~}
a_1 \eqdef n-k -\frac{5}{2} + \frac{1}{2}\sqrt{\Delta}.
\end{equation}
Let us recall that in order to have \eqref{eq:constraint3_2} satisfied, we should have $a\leq a_0$ or $a \geq a_1$. Because
of the constraint $a < n - k$ and since $a_1 > n -k$, the only case to study is  $a\leq a_0$.
Combining \eqref{eq:amin} with $a\leq a_0$, we obtain:
$$
\frac{1}{2} \Big( (1+m) n - 3 k \Big) \leq a_0.
$$
which is equivalent to the following inequality involving this time a degree-2 polynomial in $m$:
\begin{equation}
n^2 m^2 + 2n(1-n-k)m + 2kn + k^2 - 10k + n^2 - 2n \geq 0.
\end{equation}
The discriminant of this polynomial is $n^2(8k+1)$ and the roots are:
$$
m_0 \eqdef 1+ R - \frac{1}{n}  - \sqrt{ \frac{8}{n} R + \frac{1}{n^2}}
\text{~~~ and ~~~}
m_1 \eqdef  1+ R - \frac{1}{n} + \sqrt{ \frac{8}{n} R + \frac{1}{n^2}}\cdot
$$
Because of the fact that $m \leq 1 + R$ from \eqref{eq:constraint1}, and since $m_1 > 1+R$, we conclude 
that the attack can be applied as long as $m \leq m_0$, that is to say:
\begin{equation} \label{success:constraint}
m \leq 1+ R - \frac{1}{n}  - \sqrt{ \frac{8}{n} R + \frac{1}{n^2}}\cdot
\end{equation}

\item Finally, the last step of the attack consists in performing the attack of  \cite{CGGOT:dcc14}.
\end{enumerate}

\begin{remark}
This upper-bound is roughly $1+R$.
  In \cite{BBCRS11b}, the authors suggest to choose 
$m \approx 2-R$ for rates $R> \frac{1}{2}$, which is well within the reach of the present attack.
\end{remark}



\subsection{Estimating the complexity}

As explained in Proposition \ref{prop:complexity}, the 
square of a code of dimension $k$ and length $n$ can be computed in $O(n^2k^2)$.
Let us study the costs of the steps of the attack.

\begin{itemize}
\item {\bf Step 1. Finding the positions of degree $2$.}
For a constant number of subsets $\Ind$ of length $a \leq a_0$
where $a_0$ is defined in \eqref{a:limit}, we 
shorten $\Cpub^{\bot}$ and compute its square. If $a$ is close to $a_0$
then, the shortened code
has dimension $n-k-a = O(\sqrt{n})$. Hence, the computation of its square
costs $O(n^3)$. Thus this first step costs $O(n^3)$ operations in $\fq$.

\item {\bf Step 2. Transforming degree-2 positions into degree 1 positions.}
This is the most expensive part of the attack.
For a given position $i_1 \in \Jind_1$, the computation of positions $i_2$ of
degree $2$ such that\footnote{Equivalently, there exists an
integer $j$ such that $\Tm_{i_1, j} \neq 0$ and $\Tm_{i_2, j} \neq 0$.} $j(i_1) \in j(i_2)$ consists
essentially in shortening the dual public code at $i_1$ and
applying to the shortened code the first step. This costs $O(n^3)$.
Then, the application of Proposition \ref{pr:plus} to transform $i_2$
requires to proceed to at most $q$ linear combinations and, for each one,
to check whether the position became of degree $1$. Each check
has mostly the same cost as the first step, that is $O(n^3)$.
Thus, the overall cost to
reduce one position of degree $2$ is $O(n^4)$ and hence the cost
of this second step is $O(n^5)$.

\item {\bf Step 3.} According to \cite{CGGOT:dcc14},
it is in $O(n^6)$.
\end{itemize}

\section{Experimental Results}\label{sec:exp_res}
Table \ref{table} gathers experimental results obtained when the attack is programmed in Magma~V2.20-3 \cite{magma}.
The attacked parameters are taken from \cite[Tables 3 \& 4]{BBCRS11b}
The timings given  are obtained with Intel\textsuperscript{\textregistered} Xeon 2.27GHz and 72 Gb of RAM. Our programs are far from
being optimized and probably improved programs could provide better timings
and memory usage.

The running times for codes of length $346$ are below 5 hours and those
for codes of length $546$ can be a bit longer than one day. The total memory
usage remains below $100$Mb for codes of length $346$ and $500$Mb for codes
of length $546$.

\begin{table}[!h]
\begin{center}
\begin{tabular}{p{3.5cm}p{1.3cm}p{1.4cm}p{4cm}} \toprule 
 $(q,n,k,z)$            & $m$  & Step 1 & Step 2 \\ \midrule
(347, 346, 180, 1) & 1.471 & 15s &  18513s  ($\approx$5 hours)  \\
(347, 346, 188, 1) & 1.448 & 8s & 10811s  ($\approx$3 hours)  \\
(347, 346, 204, 1) & 1.402 & 10s & 8150s ($\approx$2.25 hours)  \\
(347, 346, 228, 1) & 1.332 & 15s & 9015s ($\approx$2.5 hours)  \\
(347, 346, 252, 1) & 1.263 & 36s & 10049s ($\approx$2.75 hours)  \\
(347, 346, 268, 1) & 1.217 & 3s & 14887s ($\approx$4 hours)  \\
(347, 346, 284, 1) & 1.171 & 3s & 7165s ($\approx$2 hours)  \\
\midrule
(547, 546, 324, 1) & 1.401 & 60s & 58778s ($\approx$16 hours)  \\
(547, 546, 340, 1) & 1.372 & 83s & 72863s ($\approx$20 hours)  \\
(547, 546, 364, 1) & 1.328 & 100s & 72343s ($\approx$20 hours)  \\
(547, 546, 388, 1) & 1.284 & 170s & 85699s ($\approx$24 hours)  \\
(547, 546, 412, 1) & 1.240 & 15s &   157999s ($\approx$43 hours)  \\
(547, 546, 428, 1) & 1.211 & 15s &   109970s ($\approx$30,5 hours)  \\
 \bottomrule
\end{tabular}
\end{center}
\caption{Running times}\label{table}
\end{table}

\begin{remark}
  Since the algorithms include many random choices, the identification
  of pairs $(i_1, i_2)$, where $i_1\in \Jind_1$ and $i_2\in \Jind_2$
  such that $j(i_1) \in j(i_2)$ might happen quickly or be rather long.
  This explains the important gaps between different running times.
\end{remark}

\begin{remark}
  Actually some parameters proposed in \cite{BBCRS11b} were directly
  distinguishable without even shortening. This holds for
  $(q,n,k)=(347,346, 268)$, $(q,n,k)=(347,346, 284)$ and $(q,n,k)=(547,546, 428)$
  with $m$ respectively equal to $1.217$, $1.171$ and $1.211$. This explains
  why the first step is quicker for these examples.
\end{remark}

\begin{remark}
   The examples $[346, 180]_{347}$
   and $[346, 188]_{347}$  do not satisfy \eqref{success:constraint}. However,
   they are distinguishable by shortening and squaring and the attack works
   on them. Because of some cancellation phenomenon for positions of degree 2
   which we do not control, it may happen that the upper bound in
   Proposition \ref{prop:upper_bound} is not sharp and that some shortenings
   of $\Cpub^{\bot}$
   turn out to be distinguishable while our formulas could not
   anticipate it.
\end{remark}

The above remark is of interest since it points out that our
attack might work for values of $m$ above $1+R$.

\section{Concluding Remarks}
\label{sec:conclusion}

The papers \cite{BaldiChiara07,BaldiChiara08,BBCRS11a,BBCRS11b} can be seen as an attempt of replacing 
the permutation matrix in the McEliece scheme by a more complicated 
transformation.
Instead of having as in the McEliece scheme a relation between the secret 
code $\Csec$ and the public code $\Cpub$ of the form
$\Csec = \Cpub \Pi$ where $\Pi$ is a permutation matrix, it was  chosen in \cite{BaldiChiara07,BaldiChiara08}
that 
$$
\Csec = \Cpub \Tm
$$
where $\Tm$ is a sparse matrix of density $m$ or as
$$
\Csec = \Cpub (\Tm + \Rm)
$$
where $\Tm$ is as before   
and $\Rm$ is of very small rank $z$ (the case of rank $1$ being probably the only practical way of choosing this rank as will be
discussed below) as in \cite{BBCRS11a,BBCRS11b}.
It was advocated that this allows to use for the secret code $\Csec$, codes which are well known to be weak in the usual 
McEliece cryptosystem such as LDPC codes \cite{BaldiChiara07,BaldiChiara08} or 
GRS codes  \cite{BBCRS11a,BBCRS11b}.
Interestingly enough, it turns out that for LDPC codes this basically amounts choosing a McEliece system where the density of the parity-check matrix is increased by a large
amount and the error-correction capacity is decreased by the same multiplicative constant. The latter approach has been studied in \cite{MTSB13}, it leads to schemes
with slightly larger decoding complexity but that have at least partial security proofs.

In the case of GRS codes, the first attempt \cite{BBCRS11a} of choosing for $\Tm$ a permutation matrix was broken in \cite[Sec.4]{CGGOT:dcc14}.
It was suggested later on \cite{BBCRS11b} that this attack can be avoided by choosing $\Tm$ of larger density. 
In order to reduce the public key size when compared to the McEliece scheme based on Goppa codes,
rather moderate values of $m$ between $1$ and $2$ ($m=1.4$ for instance) were chosen in \cite{BBCRS11b}. We show here that the parameters 
proposed in \cite{BBCRS11b} can be broken by a new attack computing first the dimension of the square code of shortened versions of the dual of the public code
and using this to reduce the problem to the original problem \cite{BBCRS11a} when $\Tm$ is a permutation matrix. 
This attack can be avoided by choosing larger values for $m$ and/or $z$, but this comes at a certain cost as we now show.

\par{\bf Increasing $z$.}
Increasing $z=1$ to larger values of $z$ avoids the attack given here, though some of the ideas of \cite{CGGOT:dcc14} might be used in this new context
to get rid of the $\Rm$ part in the scheme and might lead to an attack of reasonable complexity when $z=2$ by trying first to guess
several codewords which lie in the code $\CC \eqdef \Csec^\perp \Tm^T \cap \Cpub^\perp $ (this code is of codimension at least
$z$ in $\Cpub^\perp$). Once $\CC$ is found, we basically have to recover $\Tm$ and the approach used in this paper can be applied to it.
To avoid such an attack, rather large values of $z$ have to be chosen, but the decryption cost becomes prohibitive by doing so.
Indeed, decryption time is of order $q^z C$ where $C$ is the decoding complexity of the underlying GRS code.
Choosing $z=2$ is of questionable practical interest and $z>2$ becomes probably unreasonable.

\par{\bf Increasing $m$.}
Choosing values for $m$ close enough to $2$ will avoid the attack presented here.
However  this
also reduces strongly the gain in key size when compared to the McEliece scheme based on Goppa or alternant codes. 
Indeed, assume for simplicity $m=2$.
We can use in such a case for the secret code a GRS code over $\fq$ of dimension $k=n-2t$ and add errors of weight 
$\leq \frac{t}{2}$ in the BBCRS scheme. The public key size of such a scheme is however not better than choosing
in the McEliece scheme a Goppa code of the same dimension $n-2t$ but which is the subfield subcode of
a GRS code over $\F_{q^2}$ of dimension $n-t$, and which can also correct $\frac{t}{2}$ errors.
This Goppa code has the very same parameters and provides the same
security level.  For this reason, one loses the advantages of using GRS codes
when choosing $m$ close to $2$.
Thus, to have interesting key sizes and to resist to our attack $m$ should be smaller than $2$ and larger than $1+R$. One should however be careful, since, as explained in \S \ref{sec:exp_res}, it is still unclear whether the attack fails for $m$ closely above $1+R$. 

On the other hand, it might be interesting for theoretical reasons to understand better the security of the BBCRS scheme
for larger values of $m$. There might be a closer connection than what it looks between the BBCRS scheme
with density $m$ and the usual McEliece scheme with (possibly non-binary) Goppa codes of extension degree $m$.
The connection is that the case  $m=2$ is in both cases the limiting case where the distinguishing approach of
\cite{CGGOT:dcc14,COT:EC14} might work (in \cite{COT:EC14}, the attack only works because wild Goppa codes are studied
and this brings an additional power to the distinguishing attack).
It should also be added that it might be interesting to study the choice of $\Csec$ being an LDPC code and $\Csec = \Cpub (\Tm + \Rm)$ since 
here adding $\Rm$ of small rank can also change rather drastically the property of  $\Cpub$ being an LDPC code (which is at the heart of the key 
attacks on McEliece schemes based on LDPC codes).

\bibliographystyle{splncs03}
\bibliography{crypto,qubib}

G\cleardoublepage
\newpage

\appendix

\section{Proof of Lemmas  \ref{lem:punc_short_GRS}
 to \ref{lem:inclusion_punctured_shortened}}
\label{sec:A}

\begin{notation}\label{nota:D}
  In the proofs to follow, the code $\code{D}$ always denotes
the one defined in Lemma \ref{lem:puncture_1}, that is
$$
\code{D} \eqdef \Csec^\perp \Tm ^T.
$$
\end{notation}

  \begin{proof}[Proof of Lemma \ref{lem:punc_short_GRS}]
    A codeword of $\GRS{k}{\xv}{\yv}$
    is of the form $(y_1 P(x_1), \ldots, y_n P(x_n))$ where $P\in \fq[X]$
    has degree $<k$.
    Puncturing consists in removing the entries with index in $\Ind$,
    which yields
    (\ref{eq:punctGRS}). To lie in the shortened code,
    the polynomial $P$ should vanish on the elements of $\Ind$. Thus,
    it should be of the form $P(X) = (\prod_{i \in \Ind} (X-x_i))Q(X)$
    for some polynomial $Q$ of degree $<k-|\Ind|$.
    Hence, the words of the shortened code are of the form
    $$
    {\left(y_i \left(\prod_{j \in \Ind} (x_i-x_j)\right)Q(x_i) \right)}_{i \in \{1, \ldots, n\} \setminus \Ind},
    $$
    where $Q$ has degree $<k-|\Ind|$, which yields (\ref{eq:shortGRS}).
  \end{proof}

\begin{proof}[Proof of Lemma \ref{lem:puncture_1}]   For any  codeword $\cv=(c_i)_{1 \leq i \leq n}$ in $\code{D}$, there exists a polynomial $P(X) \in \fq[X]$ of 
degree less than $n-k$ such that for all $i \in \{1,\dots,n\}$, we have
\begin{eqnarray*}
c_i & = & \sum_{j=1}^n y_j P(x_j) T^T_{ji} \\
& = & \sum_{j=1}^n y_j P(x_j) T_{ij}.
\end{eqnarray*}  
When $i$ is in $\Jind_1$, we clearly have
$c_i = T_{ij(i)} y_{j(i)} P(x_{j(i)})$.
This implies that
\begin{equation}
\label{eq:easy}
\pu{\Jind_2}{\code{D}} \subseteq \GRS{n-k}{\uv}{\vv}.
\end{equation}
\end{proof}

\begin{proof}[Proof of Lemma \ref{lem:puncture_2}]
By definition of $\Csec^\perp(\lambda)$, we have $\Csec^\perp(\lambda) \Rm^T = 0$
and hence,
\begin{eqnarray*}
\Cpub^\perp(\lambda) & = & \Csec^\perp(\lambda) (\Tm^T + \Rm^T) \\
& = &\Csec^\perp(\lambda)\Tm^T  \\
& \subseteq & \code{D}.
\end{eqnarray*}
Then, using (\ref{eq:easy}), we get
\begin{eqnarray*}
\pu{\Jind_2}{\Cpub^\perp(\lambda)} & \subseteq & \pu{\Jind_2}{\code{D}} \\
& \subseteq & \GRS{n-k}{\uv}{\vv}.
\end{eqnarray*}
This is precisely the inclusion given by \eqref{eq:punctured}.

Let us now assume that $\Jind_1$ contains an information set of $\Csec^\perp$ and $\Tm^T$ is invertible.
Consider a codeword $\cv$ in $\Cpub^\perp$.
There exists $\cv'$ in $\Csec^\perp$ such that
\begin{equation}
\label{eq:one}
\cv = \cv' (\Tm^T+\Rm^T).
\end{equation}
Notice now that
\begin{eqnarray*}
\cv' \Rm^T & = & \cv' (\lambda^T \mu) \\
& = & (\scp{\cv'}{\lambda}) \mu.
\end{eqnarray*}
Let $\pv \eqdef \pu{\Jind_2}{\cv'\Tm^T}$. Since $\Jind_1$ contains an information set of $\Csec^\perp \Tm^T$ and $\Tm^T$ is invertible, 
 the composite map:
$$
\Cpub^{\bot} \longrightarrow \Csec^{\perp} \longrightarrow \Csec^{\perp}\Tm^T \longrightarrow \pu{\Jind_2}{\Csec^{\perp}\Tm^T}
$$
is an isomorphism and hence, we deduce that there exists
$\bv$ in $\fq^{n-|\Jind_2|}$ (which does not depend on $\cv'$) such that
$$
\scp{\cv'}{\lambda} = \scp{\pv}{\bv}.
$$
We define $\av$ by $\av \eqdef \pu{\Jind_2}{\mu}$ and we obtain by using
\eqref{eq:one} that
\begin{eqnarray*}
\pu{\Jind_2}{\cv} & = & \pu{\Jind_2}{\cv' \Tm^T} + \pu{\Jind_2}{\cv'\Rm^T}\\
 & = & \pv + (\scp{\pv}{\bv}) \av,
\end{eqnarray*}
which proves \eqref{eq:further}.
\end{proof}

\begin{proof}[Proof of Lemma \ref{lem:inclusion_punctured_shortened}]
We start by the remark  that there exists a vector $\av_0 \in \fq^n$ such that
$$
\Cpub^{\bot} = \Cpub(\lambda)^{\bot} + <\av_0>.
$$
Now, after shortening at $\Ind_1$, there exists $\av_1 \in \fq^{n-|\Ind_1|}$
such that
$$
\sh{\Ind_1}{\Cpub^{\bot}} = \sh{\Ind_1}{\Cpub(\lambda)^{\bot}} + <\av_1>
$$
and finally there exists a vector $\av$ in $\fq^{n-s-|\Jind_2|}$ such that
 \begin{equation}
 \label{eq:starting_point}
 \pu{\Jind_2}{\sh{\Ind_1}{\Cpub^\perp}} \subseteq \pu{\Jind_2}{\sh{\Ind_1}{\Cpub^\perp(\lambda)}} + <\av>. 
 \end{equation}
Moreover, we have, 
 \begin{eqnarray}
 \pu{\Jind_2}{\sh{\Ind_1}{\Cpub^\perp(\lambda)}} & = & \pu{\Jind_2}{\sh{\Ind_1}{\Csec^\perp(\lambda) (\Tm^T+\Rm^T)}} \nonumber\\
 & = &  \pu{\Jind_2}{\sh{\Ind_1}{\Csec^\perp(\lambda) \Tm^T}} \nonumber\\
 & \underset{\textrm{Codim}\ 1}{\subseteq} & \pu{\Jind_2}{\sh{\Ind_1}{\code{D}}}, \label{eq:punctured_D}
 \end{eqnarray}
where we remind that $\code{D}$ is defined in Notation~\ref{nota:D}.
 From the definition of $\code{D}$ we know that
 $$
 \code{D} = \left\{ \left. \left(\sum_{j=1}^n y_j T_{ij} P(x_j) \right)_{i=1}^n 
    ~ \right| ~ \deg P < n-k\right\}.
 $$
 Observe that for a position $i \in \Jind_1$, we have 
 \begin{equation}
 \label{eq:fundamental}
 \sum_{j=1}^n y_j T_{ij} P(x_j) = y_{j(i)}T_{ij(i)} P(x_{j(i)}).
 \end{equation}
 Such a coordinate vanishes if and only if $P(X)$ is divisible by  $(X-x_{j(i)})$ which implies that
 $$
 \sh{\Ind_1}{\code{D}} = \left\{ \left. \left(\sum_{j=1}^n y_j T_{ij} \prod_{l \in \Ind_1}(x_j-x_{j(l)})  P(x_j) \right)_{i \notin \Ind_1} ~ \right| ~ \deg P < n-k-s\right\}.
 $$
 From this and using \eqref{eq:fundamental} again, we obtain 
 $$
 \code{E}_0 \eqdef \pu{\Jind_2}{\sh{\Ind_1}{\code{D}}} =
\left\{ \left. \left(y_{j(i)} T_{ij(i)} \prod_{l \in \Ind_1}(x_{j(i)} -x_{j(l)})  P(x_{j(i)}) \right)_{i \in \Jind_1 \setminus \Ind_1} ~\right|~ \deg P < n-k-s\right\} \cdot
 $$
 This is clearly a GRS code of degree $n-k-s$.
Set 
\begin{equation}\label{eq:codeE}
\code{E} \eqdef \pu{\Jind_2}{\sh{\Ind_1}{\Cpub(\lambda)^{\perp}}}.
\end{equation}
Then, $\code{E}$ is indeed a subcode of the GRS code
$\code{E}_0$ of codimension $1$
and the lemma follows by combining this equation with
\eqref{eq:starting_point} and \eqref{eq:punctured_D}
and using that the left-hand term in \eqref{eq:punctured_D} has codimension
$1$ in the right hand one.
\end{proof}

\section{Proof of Proposition  \ref{prop:upper_bound}}
\label{sec:upper_bound}

To prove this result we will need a few additional results involving  general inequalities concerning the dimension of
the square code.

\begin{lemma}
\label{lem:super_easy}
For all linear codes $\AC,\BC,\CC \in \fq^n$ and all subsets $\Ind$ of code positions,  we have
\begin{eqnarray}
\dim \sqc{\left(\AC+\BC\right)} & \leq & \dim \sqc{\AC} + \dim \sqc{\BC} + \dim \left( \stp{\AC}{\BC}\right)
\label{eq:sum}\\
\dim \sqc{\CC} & \leq & \dim \sqc{\pu{\Ind}{\CC}} + |\Ind| \label{eq:sq_puncture}.
\end{eqnarray}
\end{lemma} 

\begin{proof}
By definition of the square codes, one proves that 
$
(\code{A}+\code{B})^2 = \code{A}^2 + \code{B}^2 + \code{A} \star \code{B}.
$ This leads to (\ref{eq:sum}). Let $\code{A}(\Ind)$ be the code
of dimension $|\Ind|$ and length $n$ composed by all words of $\fq^n$
supported by $ \Ind$ and ${\rm Ext}(\pu{\Ind}{\CC})$ be the code
$\pu{\Ind}{\CC}$ extended by zero to get a code of length $n$. We have:
$$
\code{C} \subseteq {\rm Ext}(\pu{\Ind}{\CC}) \oplus \code{A}(\Ind),
$$
then, thanks to (\ref{eq:sum}) and since
$\stp{{\rm Ext}(\pu{\Ind}{\CC})}{\code{A}(\Ind)}=\{0\}$ and $\dim \sqc{\code{A}(\Ind)} =
|\Ind|$, we get (\ref{eq:sq_puncture}).
\end{proof}

\begin{proof}[Proof of Proposition \ref{prop:upper_bound}]
We start by using Lemma \ref{lem:inclusion_punctured_shortened}
with $\Ind  = \Ind_1$. We get that 
there exist vectors $\av,\uv,\vv$ in  $\fq^{n-|\Ind|-|\Jind_2|}$ such that 
\begin{equation}
\label{eq:1}
\pu{\Jind_2}{\sh{\Ind_1}{\Cpub^\perp}} \subseteq \code{E} + <\av>,
\end{equation}
where $\code{E}$ is some subcode of $\GRS{n-k-|\Ind_1|}{\uv}{\vv}$ of codimension $1$ (see (\ref{eq:codeE})). 
From \eqref{eq:sum} 
with $\code{A} = <\av>$ and $\code{B}=\code{E}$ and
Proposition \ref{prop:square}, we deduce that
\begin{eqnarray}
\dim \sqcp{\pu{\Jind_2}{\sh{\Ind_1}{\Cpub^\perp}}} &\leq & 1 + \dim \sqcp{\code{E}} + \dim \code{E}\\
& \leq & 1 + \dim \sqcp{\GRS{n-k-|\Ind_1|}{\uv}{\vv}} +(n-k-|\Ind_1|-1)\\
& \leq & 1 + 2(n-k-|\Ind_1|)-1 + (n-k-|\Ind_1|-1)\\
& = & 3(n-k)-3|\Ind_1|-1.
\end{eqnarray}
We finish by using \eqref{eq:sq_puncture} and obtain that
\begin{eqnarray*}
\dim \sqcp{\sh{\Ind_1}{\Cpub^\perp}} & \leq & \dim \sqcp{\pu{\Jind_2}{\sh{\Ind_1}{\Cpub^\perp}}} + |\Jind_2| \\
& \leq & 3(n-k)-3|\Ind_1|-1 + |\Jind_2|.
\end{eqnarray*}
\end{proof}

\section{Explanations for the upper bound \eqref{eq:upper_bound} in the case where $\Ind$ contains
positions of degree $2$}\label{sec:heuristic}

\subsection{A general upper-bound on the dimension of $\sqcp{\sh{\Ind}{\Cpub^\perp}}$}

It will be useful to have a slight variant of Proposition \ref{prop:upper_bound} which holds for any subset $\Ind$ of code positions and which is
 given by
\begin{proposition}\label{prop:general_upper_bound}
Let $\Ind$ b a set of code positions.
Let $\CC \eqdef \sh{\Ind}{\Csec^{\perp} \Tm^T}$ and ${\rm Ext}(\CC)$  be the code obtained from $\CC$ by extending 
by zero at positions which were shortened.
 If ${\rm Ext}(\CC) \nsubseteq \left(\Csec(\lambda)^{\
    \perp} \right) \Tm^{T}$
and $\Tm$ is invertible,
we have
\begin{equation}
\label{eq:general_upper_bound}
\dim \sqcp{\sh{\Ind}{\Cpub^\perp}} \leq \dim \sqc{\CC} + \dim  \CC.
\end{equation}
\end{proposition}
\begin{proof}
Recall that, from Lemma \ref{lem:basics}, we have
$$
\Cpub^\perp = \Csec^{\perp}(\Tm^T + \Rm^T)
$$
and that
$\Rm^T = \lambda^T \mu$.
Two cases have to be considered now.

\noindent
Case 1: $\lambda \in  \Csec$. This implies that 
$\Csec^\perp \subset <\lambda>^\perp$ and therefore we have that
\begin{eqnarray*}
\Cpub^{\perp} & = & \Csec^\perp(\Tm^T+ \lambda^T \mu)\\
& = & \Csec^\perp \Tm^T.
\end{eqnarray*}
In such a case we have
$
\sh{\Ind}{\Cpub^\perp} = \CC
$
and the upper-bound follows immediately.

\noindent
Case 2: $\lambda \notin  \Csec$.
In this case there exists $\av \in \fq^n$ such that
$$
\Csec^\perp = <\av> + \Csec(\lambda)^\perp.
$$
From this we deduce
\begin{eqnarray*}
\Cpub^\perp & = & \left(<\av> + \Csec(\lambda)^\perp \right)(\Tm^T+ \lambda^T \mu)\\
& = & <\bv> + \Csec(\lambda)^\perp  \Tm^T,
\end{eqnarray*}
where $\bv= \av (\Tm^T+ \lambda^T \mu)$.
This implies that there exists $\cv \in \fq^{n-|\Ind|}$ such that
$$
\sh{\Ind}{\Cpub^\perp} \subseteq <\cv> + \sh{\Ind}{\Csec(\lambda)^\perp \Tm^T}.
$$
Then, we use the upper bound \eqref{eq:sum} of Lemma \ref{lem:super_easy} 
with $\code{A} = <\cv>$ and $\code{B}=\sh{\Ind}{\Csec(\lambda)^\perp \Tm^T}$ and obtain
$$
\dim \sqc{\sh{\Ind}{\Cpub^\perp}} \leq  1 + \dim \sqc{\code{B}} + \dim \code{B}.
$$
We finish the proof by noticing that 
\begin{eqnarray*}
\sqc{\code{B}} &=& \sqcp{\sh{\Ind}{\Csec(\lambda)^\perp \Tm^T}}\\
&\subseteq & \sqcp{\sh{\Ind}{\Csec^\perp  \Qm^T}}=\sqc{\CC}.
\end{eqnarray*}
Moreover, 
$\dim \code{B} \leq  \dim \CC -1$
when $\lambda \notin  \Csec$. Indeed,
since $\Tm$ is assumed to be invertible and, by assumption,
$\lambda \notin \Csec$, the code
$\Csec(\lambda)^{\perp}\Tm^T$ has codimension $1$
in $\Csec^{\perp}\Tm^T$.
Second, notice that the code $\code{B}$ extended by zero equals
$ \Csec(\lambda)^{\perp}\Tm^T  \cap {\rm Ext}(\CC)$.
Next, since  by assumption 
${\rm Ext}(\CC) \nsubseteq \Csec(\lambda)^{\perp} \Tm^{T}$,
we get that $\Csec(\lambda)^{\perp}\Tm^T  \cap {\rm Ext}(\CC)$
has codimension $1$ in $\Csec^{\perp}\Tm^T \cap {\rm Ext}(\CC) = {\rm Ext}(\CC)$.
Therefore $\dim \code{B} = \dim \code{C} - 1$.
\end{proof}

\subsection{The graph associated to a sparsely mixed GRS code}
\label{sec:graph_smGRS}
Proposition \ref{prop:general_upper_bound} raises the issue of understanding the structure
of shortenings of codes of the form $\sh{\Ind}{\CC \Tm^T}$\footnote{To keep the connection with the BBCRS scheme and the meaning 
of $\Tm$ in the BBCRS scheme we keep the transpose throughout the paper.} where 
$\CC$ is a GRS code, $\Tm$ is a sparse square matrix and $\Ind$ is a subset of code positions of $\CC$. We denote codes of this kind {\em shortened sparsely mixed GRS codes}. 
We will represent them by their defining triple $(\CC,\Tm,\Ind)$. A colored graph associated to the pair $(\Tm,\Ind)$ will be very useful for studying such codes.
It is defined as follows.

\begin{definition}[graph associated to a shortened sparsely mixed GRS code $(\CC,\Tm,\Ind)$]
The graph associated to the shortened sparsely mixed GRS code $(\CC,\Tm,\Ind)$ is the bipartite graph
$G(U\cup V,E)$ given by
\begin{itemize}
\item set of vertices $U \cup V$, with a bijection from $U$ to the set of columns of 
$\Tm$ and where $V$ is in bijection with the rows of $\Tm$.
\item edge set $E$ where $uv$ is an edge of $E$ if and only $T_{vu} \neq 0$.
\end{itemize}
All the vertices are colored in black with the exception of the vertices of $V$ which belong to $\Ind$, in such 
a case they are colored in red and the vertices of $V$ which are of degree $1$ but which do not belong to $\Ind$ are colored in blue.
\end{definition} 

\begin{remark}
 Notice that for the graph associated to the triple $(\Csec^\perp,\Tm,\Ind)$ 
(where $\Ind$ is arbitrary) corresponding to a BBCRS scheme, the positions of degree $2$ in the BBCRS scheme correspond precisely to the vertices of $V$ of degree $2$ whereas
the positions of degree $1$ in the BBCRS scheme correspond to  the vertices of $V$ of degree $1$.
\end{remark}

This graph (without the coloring) is a way of representing a code of the form $\GRS{n-k}{\xv}{\yv}\Tm^T$.
Consider a codeword $(c_1,\dots,c_n)$ of such a code. Clearly there exists a polynomial in 
$\fq[X]$ of degree $< n-k$ such that for any $v \in \{1,\dots,n\}$ we have
$$
c_v = \sum_{u \sim v} y_{uv}P(x_u)
$$
where $y_{uv} \eqdef y_u T_{vu}$ and the sum is taken over all vertices $u$ of $U$ that are adjacent to $v$. To understand the effect of shortening the code in a set of positions, consider a codeword
$(c_v)_{v \in V \setminus \Ind}$ of the shortened code. The coordinates of such a codeword are given by the same formula as before, i.e. 
$c_v = \sum_{u \sim v} y_{uv}P(x_u)$ with $P$ a polynomial of degree $< n-k$, but now this polynomial should satisfy all the equations
$$
 \sum_{u \sim i} y_{uv}P(x_u) = 0
$$ for all the $i$'s that belong to $\Ind$ (and which are therefore colored in red).

 From now on and for the rest of the section we will assume that
 \begin{hypothesis}
  The set of positions of degree $1$
of the sparsely mixed GRS code $\GRS{n-k}{\xv}{\yv}\Tm^T$ contains an information set of the code and there
are no positions of degree $>2$.
\end{hypothesis}
The effect of shortening on the dimension of $\dim \CC^2$ when $\CC = \sh{\Ind}{\GRS{n-k}{\xv}{\yv}\Tm^T}$ is better understood when we study first some extremal cases 
\begin{enumerate}
\item $\Ind$ contains only positions of degree $1$;
\item $\Ind$ is reduced to a single degree $2$ position.
\end{enumerate}

{\bf Note}: Recall that $\Jind_1, \Jind_2$ denote respectively
the sets of degree $1$ and $2$ positions associated to the sparsely
mixed code $\GRS{n-k}{\xv}{\yv}\Tm^T$ and $\CC$ denotes the shortened
code $\sh{\Ind}{\GRS{n-k}{\xv}{\yv}\Tm^T}$.

{\em Shortening with respect to positions of degree $1$. } 
In such a case, by using  Inequality \eqref{eq:sq_puncture} from Lemma \ref{lem:super_easy} with $\Ind=\Jind_2$, i.e. we expect that 
$$
\dim \sqc{\CC} \leq \sqc{\pu{\Jind_2}{\CC}} + |\Jind_2|.
$$
The point with this puncturing is that $\pu{\Jind_2}{\CC}=\pu{\Jind_2}{\sh{\Ind}{\GRS{n-k}{\xv}{\yv}}}$ is a GRS code of dimension $n-k-|\Ind|$ and we can apply Proposition \ref{prop:square}  to obtain that 
$\dim \sqc{\pu{\Jind_2}{\CC}} = 2(n-k-|\Ind|)-1$. This yields the upper-bound
\begin{equation}\label{eq:degree1}
\dim \sqc{\CC} \leq  2(n-k)-1-2|\Ind|+|\Jind_2|
\end{equation}
It turns out that this upper-bound is typically met.

{\em Shortening with respect to a single position $i$ of degree $2$: $\Ind=\{i\}$.} We would typically expect that in such a case  to have
$$
\dim \sqc{\CC} \leq  2(n-k)-2+|\Jind_2|
$$
since $|\Jind_2|$ drops by $1$.

It turns out that a stronger inequality holds in this case, the point being that shortening in a degree $2$ position $i$ 
yields codewords of the form $(\sum_{u \sim v} y_{uv} P(x_u))_{v \neq i}$ where $P$ is a polynomial 
of degree $< n-k$ which is such that $y_{j_1i} P(x_{j_1}) + y_{j_2i} P(x_{j_2})=0$ where 
the two vertices of $U$ adjacent to $i$ are $j_1$ and $j_2$ respectively. In such a case it turns out that
\begin{proposition}
\label{prop:degree2}
Assume that $k$ is an integer in the range $[\frac{n-1}{2},n-4]$.
Let $a,b,\lambda$  be three elements in $\fq$ with $a \neq b$, $\lambda \neq 0$ and 
let $\code{F}$ be the subcode of $\GRS{n-k}{\xv}{\yv}$ given by
$$
\code{F} \eqdef \left\{(y_i P(x_i))_{1 \leq i \leq n}|\deg P < n-k, P(a)=\lambda P(b) \right\}
$$
Then $\sqc{\code{F}}$ is a subcode of codimension $1$ of the GRS code
$\GRS{2(n-k)-1}{\xv}{\yv \cwp \yv}$ given by
$$
\sqc{\code{F}} = \left\{(y_i^2 P(x_i))_{1 \leq i \leq n}|\deg P < 2(n-k)-1, P(a)=\lambda^2 P(b) \right\}
$$
\end{proposition}

\begin{proof}
$\sqc{\code{F}}$ is generated by codewords of the form 
$(y_i ^2P(x_i)Q(x_i))_{1 \leq i \leq n}$ with 
\begin{eqnarray*}
\deg P & < & n-k\\
\deg Q & < & n-k\\
P(a) & = & \lambda P(b)\\
Q(a) & = & \lambda Q(b).
\end{eqnarray*}
In other words, we proved that
$$
\sqc{\code F} \subseteq \left\{(y_i^2 P(x_i))_{1 \leq i \leq n}|\deg P < 2(n-k)-1, P(a)=\lambda^2 P(b) \right\}
$$
and to conclude the proof, we only need to prove that $\dim \sqc{\code{F}}
\geq 2(n-k)-2$.

Let $\code{F}_0 \subseteq \code{F}$ be defined as
$$
\code{F}_0 \eqdef \left\{ {(y_iP(x_i))}_{1\leq i\leq n}|
\deg P <n-k,\ P(a)=P(b)=0\right\}.
$$
and let $P, Q_a, Q_b$ be $3$ polynomials of degree $<n-k$ such that 
$P(a) = \lambda P(b)\neq 0$, $Q_a$ vanishes with multiplicity $2$ at $a$ 
and $1$ at $b$ and $Q_b$ vanishes with multiplicity $1$ at $a$ and $2$ at $b$.
The existence of such nonzero polynomials is guaranteed since we assumed
$k$ to be $\leq n-4$ and hence $n-k \geq 4$.
For the very same reason, $\code{F}_0$ is nonzero.
The code $\code{F}_0$ is a GRS code whose square is the subcode of
$\GRS{2(n-k)-1}{\xv}{\yv}$ of codimension $4$ corresponding to polynomials
vanishing at $a$ and $b$ with multiplicity $2$ at both points. 
Thus, $\dim \sqc{\code{F}_0} = 2(n-k)-5$.
Next, one sees easily that
$$
\sqc{\code{F}}_0 \oplus <PQ_a> \oplus <PQ_b> \oplus <P^2> \subseteq \sqc{\code{F}},
$$
and hence $\sqc{\code{F}}$ has dimension at least $2(n-k)-2$. 
This concludes the proof.
\end{proof}

By applying this proposition to $\pu{\Jind_2 \setminus \{i\}}{\CC}$, we easily obtain the stronger inequality
$$
\dim \sqc{\CC} \leq  2(n-k)-3+|\Jind_2|,
$$
which is nothing but \eqref{eq:degree1} with $|\Ind| = 1$.

Generalizing this idea to general subsets $\Ind$ of code positions we expect that \eqref{eq:degree1} holds in all cases. However it turns out that we have in general
a stronger inequality. This is due to the fact that shortening  has sometimes long range effects. This is a point that we study in more detail in the following subsection.

\subsection{Shortening induces merging and pruning of the graph.}\label{ss:shortening_merging_pruning}
We first start with a few examples which will help to understand the underlying phenomena.

\begin{example}\label{ex:ex1}
Let us shorten in a single position, i.e. $\Ind=\{v\}$ with $v$ of degree $1$ adjacent to $u$ and we assume that there exists $v'$ of degree $2$ adjacent to $u$ and $u'$
(see Figure \ref{fig:ex1}).
\begin{figure}[h!]
\caption{A first example \label{fig:ex1}}
\centering
\includegraphics[height=2.5cm]{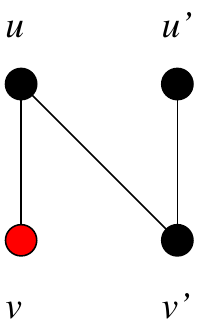}
\end{figure}
When we shorten $\GRS{n-k}{\xv}{\yv}\Tm^T$ in $v$ this means that we consider only codewords of the form
$(\sum_{u \sim w} y_{uw} P(x_w))_{w \neq v}$ where $\deg P < n-k$ and $P(x_u)=0$. Notice now that because of this,
the code position $v'$ simplifies from $y_{uv'} P(x_u) + y_{u'v'} P(x_{u'})$ to $y_{u'v'} P(x_{u'})$. In this sense, 
the shortened sparsely mixed GRS code $\sh{\{v\}}{\GRS{n-k}{\xv}{\yv} \Tm^T}$ corresponds to a subcode of the
sparsely mixed GRS code associated to a simplified graph.
It is obtained from the  graph associated to the shortened sparsely mixed code by removing
$u$ and $v$ and the incident edges. Moreover its  codewords $(\sum_{u \sim w} y_{uw} P(x_u))_{w \in V \setminus \{v\}}$
correspond to polynomials satisfying an additional condition, namely $P(x_u)=0$. 
\begin{figure}[h!]
\caption{The simplified graph corresponding to the first example \label{fig:ex1s}}
\centering
\includegraphics[height=2.5cm]{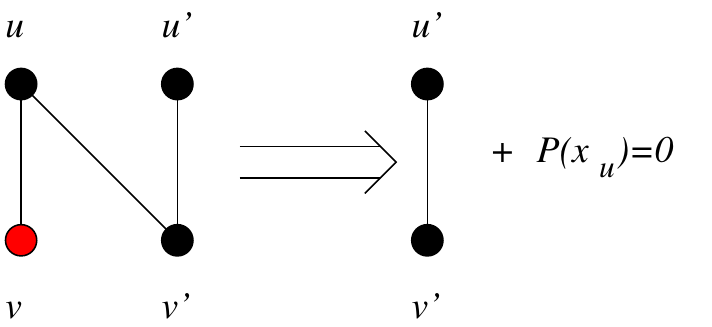}
\end{figure}
\end{example}

 In the new code, a degree $2$ vertex disappeared and we therefore expect that the dimension of the square of the shortened sparsely mixed code is equal to
 \begin{eqnarray*}
 \dim\sqcp{\sh{\{v\}}{\GRS{n-k}{\xv}{\yv} \Tm^T}}& =& \dim \sqc{\pu{\Jind_2 \setminus \{v\}}{\sh{\{v\}}{\GRS{n-k}{\xv}{\yv} \Tm^T}}} + |\Jind_2 \setminus \{v\}|\\
 & = &  2(n-k)-1 -2 + |\Jind_2| -1= 2(n-k)-4+|\Jind_2|.
 \end{eqnarray*}

 In other words, we have to take into account that the effect of shortening may have deeper effects than just the sum of the effects of the shortening of degree $1$ positions
 and degree $2$ positions which decreases the dimension by a term which is $|\Ind|$. As shown by this example, shortening might remove some other degree $2$ positions which were not 
 shortened and which could be transformed into a degree $1$ position as is apparent from this example. We therefore expect that the effect of shortening in a set $\Ind$ leads  to a dimension 
 for $\sqc{\CC}$ which is of the form
 $$
 \dim \CC = 2(n-k)-1 - |\Ind| + |\Jind'_2|
 $$
 where $\Jind'_2$ is the set of code positions of degree $2$ which remain after we take into account the effect of the shortening.
Next, it turns out that we have to take into account a slightly more complicated
 phenomenon coming from the effect of the shortening of degree $2$ positions.
This is illustrated by the next example.
 
 \begin{example}
 Let us consider now an example where we shorten in two positions $v$ and $v'$ whose neighborhood is specified by Figure 
 \ref{fig:ex2}. 
 \begin{figure}[h!]
\caption{A second example \label{fig:ex2}}
\centering
\includegraphics[height=2.5cm]{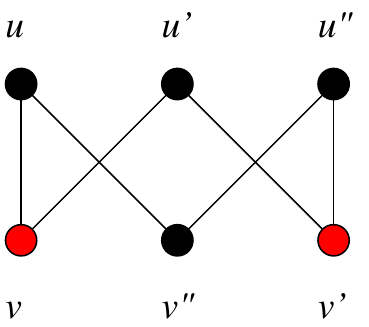}
\end{figure}
 \end{example}
 A codeword of the shortened code $\sh{\{v,v'\}}{\GRS{n-k}{\xv}{\yv}\Tm^T}$ is of the form $(\sum_{u \sim w} y_{uw} P(x_w))_{w \in V \setminus \{v,v'\}}$ where $P$ should satisfy  at the same time
 \begin{eqnarray*}
 \deg P & < & n-k\\
 y_{uv} P(x_u) + y_{u'v} P(x_{u'}) & = & 0\\
  y_{u'v'} P(x_{u'}) + y_{u"v'} P(x_{u"}) & = & 0.
 \end{eqnarray*}
 Because of these relations, the codeword position $c_{v"}$ which has not been shortened is of the form
 $$
 c_{v"} = y_{uv"} P(x_{u}) + y_{u"v"} P(x_{u"})=\alpha P(x_u)
 $$
 for some $\alpha$ that depends on the $y_{uv}$'s. In other words, the codeword position $v"$ becomes a position of degree $1$ after shortening and
 it makes sense to merge the nodes $u,u'$ and $u"$ to represent the fact that we have linear relations between $P(x_u),P(x_{u'})$ and $P(x_{u"})$, 
 see Figure \ref{fig:ex2s}. 
 \begin{figure}[h!]
\caption{The second example revisited \label{fig:ex2s}}
\centering
\includegraphics[height=2.5cm]{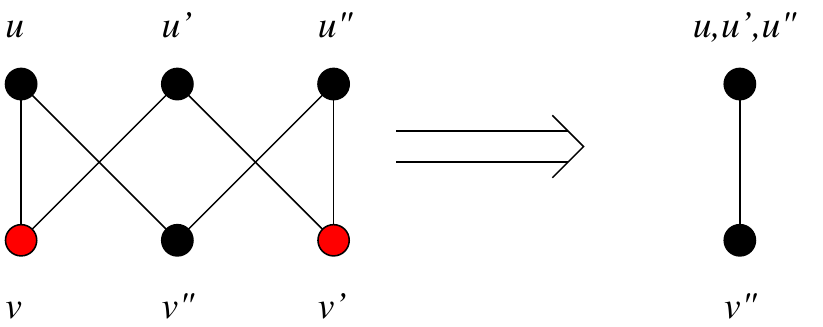}
\end{figure}

The effect on the dimension can be understood by using Inequality \eqref{eq:sq_puncture} of Lemma \ref{lem:super_easy}   and we see that we should have
\begin{eqnarray*}
 \dim\sqcp{\sh{\{v,v'\}}{\GRS{n-k}{\xv}{\yv} \Tm^T}} & \leq & \dim \sqcp{\pu{\Jind_2 \setminus \{v,v',v"\}}{\sh{\{v,v'\}}{\GRS{n-k}{\xv}{\yv} \Tm^T}}} + |\Jind_2 \setminus \{v,v',v"\} |\\
 & \leq & 2(n-k)-1 -2 + |\Jind_2| -3= 2(n-k)-6+|\Jind_2|,
 \end{eqnarray*}
 where $\Jind_2$ is the set of degree $2$ positions of the sparsely mixed code $\GRS{n-k}{\xv}{\yv} \Tm^T$.
 The $-2$ which follows the term $2(n-k)-1$ is due to the fact that the code 
 $\code{G} \eqdef \sh{\{v,v'\}}{\GRS{n-k}{\xv}{\yv} \Tm^T}$ is a code which satisfies
 $$
 \code{G} =  \left\{(y_{u(w)w} P(x_{u(w)})_{w \in V\setminus\{v,v'\}}|\deg P < n-k, P(x_u)=\alpha P(x_{u'})=\beta P(x_{u"}) \right\}
 $$
 where $u(w)$ is the unique vertex of $U$ adjacent to $w$, $\alpha$ and $\beta$ are nonzero elements of $\fq$. 
 A generalization of Proposition \ref{prop:square} leads immediately to
 $$
 \dim \sqcp{\code{D}} \leq 2(n-k)-1-2.
 $$
 
 In other words, we can quantify the effect of the shortening of degree $2$ positions by merging the vertices of $U$ which are linked to a same
 vertex of $V$ which is shortened. If we obtain a vertex which corresponds to the merging of $d$ vertices then this induces a drop in dimension of 
 $d-1$ (this corresponds to a generalization of Proposition \ref{prop:square}).
 
 All these considerations lead to introduce the following algorithm that formalizes these considerations.
 
\begin{algorithm}[h!]
Algorithm for reducing the graph after shortening.
\begin{algorithmic}
\STATE{}\COMMENT{Merge phase}
\FORALL{red nodes $v$ of degree $2$}
\STATE{Remove $v$ and the two edges $u_1 v$ and $u_2v$ incident to it.}
\STATE{Merge $u_1$ and $u_2$.}
\ENDFOR
\STATE{}\COMMENT{Pruning phase}
\WHILE{there is a red node $v$ in $V$ of degree $1$ adjacent to a black node $u$ in $U$}
\STATE{Remove $v$.}
\IF{there exists a black node $v'$ adjacent to $u$ which is of degree $1$}
\STATE{Remove $v'$ and its incident edge.}
\ENDIF
\STATE{Remove $u$ and all the edges adjacent to $u$.}
\ENDWHILE

\end{algorithmic}
\end{algorithm}

With the help of this algorithm we can bring in the crucial quantities which govern the dimension of $\sqc{\CC}$ and, from Proposition \ref{prop:general_upper_bound}, also 
$\sqcp{\sh{\Ind}{\Cpub^\perp}}$
\begin{itemize}
\item the set ${\cal M}$ of merged nodes in the graph which did not disappear during the pruning process.
\item the degree $d(x)$ of such a merged node $x$ is defined as the number of vertices of $U$ that have been merged
together to yield this node $x$.
\item the remaining set $\Jind_2'$ of degree $2$ nodes of $V$ after merging and pruning.
\item the set $\Ind_1$ of degree $1$ nodes of $V$ in the original graph that have disappeared during the process.
\end{itemize}

The dimension of $\dim \sqc{\CC}$ is typically given by
$$
\dim \sqc{\CC} = 2(n-k) -1 -2|\Ind_1| + |\Jind'_2| -\sum_{x \in {\cal M}} (d(x)-1)
$$
and from Proposition \ref{prop:general_upper_bound}  (whose upper-bound is actually generally met) the dimension of $\dim \sqcp{\sh{\Ind}{\Cpub^\perp}}$ is typically given by
\begin{equation}\label{eq:du_flan?}
\dim \sqcp{\sh{\Ind}{\Cpub^\perp}} = 3(n-k) - 1- 2|\Ind_1| + |\Jind'_2|-|\Ind|-\sum_{x \in {\cal M}} (d(x)-1).
\end{equation}

\subsection{An example}
 We give in Figure \ref{fig:example}  an example of a graph associated to a shortened sparsely mixed GRS code of length $10$ where we shortened $4$ positions.
\begin{figure}[h!]
\caption{An example of a graph associated to a shortened sparsely mixed GRS code \label{fig:example}}
\centering
\includegraphics[height=4cm]{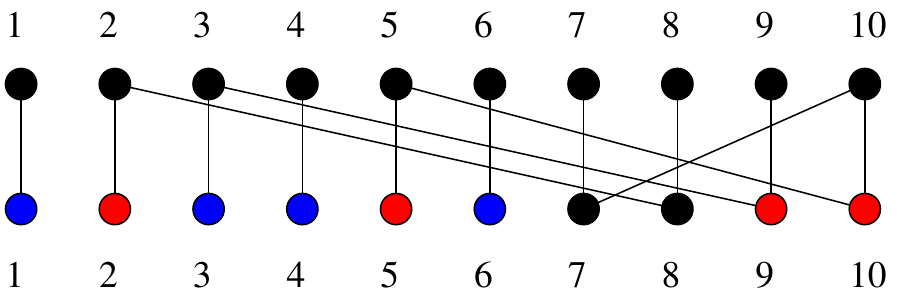}
\end{figure}
After the merging step, the graph is transformed into the graph given in Figure \ref{fig:example2}.

\begin{figure}[h!]
\caption{The graph obtained after the merging step \label{fig:example2}}
\centering
\includegraphics[height=4cm]{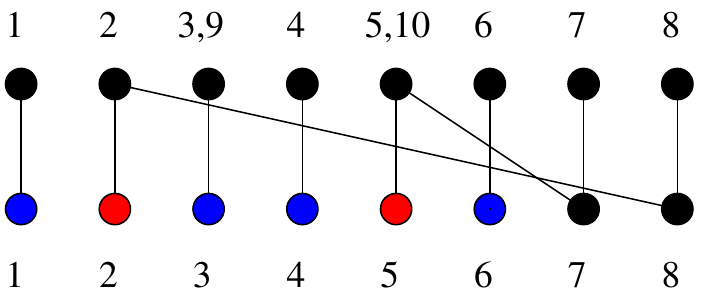}
\end{figure}
After the pruning step the graph further simplifies and becomes the graph given in Figure \ref{fig:example3}.
\begin{figure}[h!]
\caption{The graph obtained after the merging and the pruning step (here red vertices and edges do not belong to the graph, they are just here
to indicate which edges and nodes have been removed). \label{fig:example3}}
\centering
\includegraphics[height=4cm]{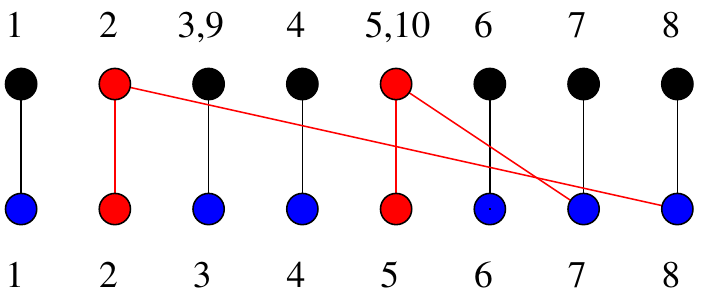}
\end{figure}

In this case
\begin{itemize}
\item ${\cal M}$ is a set of two merged nodes : $\{3,9\}$ and $\{5,10\}$;
\item the degree of both of these merged nodes is equal to $2$;
\item there remains no node of degree $2$ in $V$ after merging and pruning : $\Jind'_2 = \emptyset$.
\item two vertices of $V$ of degree $1$ have disappeared during the process $|\Ind_1|=2$.
\end{itemize}
If we assume that the dimension of the underlying GRS code was $6$ at the beginning we therefore expect a dimension of  $\dim \sqcp{\sh{\Ind}{\Cpub^\perp}}$
of
$$
\dim \sqcp{\sh{\Ind}{\Cpub^\perp}} = 3\times 6 - 1 -2\times 2 + 0 - 4 - 2 = 7
$$

\subsection{The relationship between $\dim \sqcp{\sh{\Ind}{\Cpub^\perp}}$ and $\dim \sqcp{\pu{i}{\sh{\Ind}{\Cpub^\perp}}}$}

A quick inspection of the reasoning underlying  Formula \eqref{eq:du_flan?} for $\dim \sqcp{\sh{\Ind}{\Cpub^\perp}}$ also shows that we expect that 
\begin{eqnarray*}
\dim \sqcp{\pu{i}{\sh{\Ind}{\Cpub^\perp}}} & =& 3(n-k) - 1- 2|\Ind_1| + |\Jind'_2|-|\Ind|-\sum_{x \in {\cal M}} (d(x)-1)\;\;\text{if $i \notin \Jind'_2$}\\
\dim \sqcp{\pu{i}{\sh{\Ind}{\Cpub^\perp}}} &= & 3(n-k) - 1- 2|\Ind_1| + |\Jind'_2|-1- |\Ind|-\sum_{x \in {\cal M}} (d(x)-1)\;\;\text{if $i \in \Jind'_2$}
\end{eqnarray*}

In other words we expect that 
\begin{eqnarray*}
\dim \sqcp{\sh{\Ind}{\Cpub^\perp}} - \dim \sqcp{\pu{i}{\sh{\Ind}{\Cpub^\perp}}} &= &0 \;\;\text{if $i \notin \Jind'_2$}\\
\dim \sqcp{\sh{\Ind}{\Cpub^\perp}} - \dim \sqcp{\pu{i}{\sh{\Ind}{\Cpub^\perp}}} &= &1 \;\;\text{if $i \in \Jind'_2$}
\end{eqnarray*}
The positions of degree $2$ that we detect by computing $\dim \sqcp{\sh{\Ind}{\Cpub^\perp}} - \dim \sqcp{\pu{i}{\sh{\Ind}{\Cpub^\perp}}}$
are therefore the elements of $\Jind'_2$, that is the vertices which remain of degree $2$ after merging and pruning the graph.

\section{Proof of Proposition \ref{pr:plus}}

\label{sec:plus}
We first notice that 
\begin{eqnarray*}
\CC & = & \Cpub^\perp \Dm_{\alpha,i_1,i_2}\\
& = & \Csec^\perp(\Tm^T\Dm(\alpha,i_1,i_2)
+\Rm^T\Dm(\alpha,i_1,i_2)
) \end{eqnarray*}
Note that $\Rm^T\Dm(\alpha,i_1,i_2)$ and its transpose are both of rank at most one since 
$\Rm$ is of rank $\leq 1$.
Let
$$
\Sm^T \eqdef \Tm^T\Dm(\alpha,i_1,i_2)
$$
Denote the entry in row $i$ and column $j$ of $\Tm,\Tm^T,\Sm^T,$ and $\Dm(\alpha,i_1,i_2)$ by 
$T_{ij},T^T_{ij},S^T_{ij}$, and $D(\alpha,i_1,i_2)_{ij}$
 respectively.
 From the very definition of $\Dm(\alpha,i_1,i_2)$, $\Sm^T$ and $\Tm^T$ coincide in all entries
 with the exception of the entry in row $j_1$ and column $i_2$ where we have
 \begin{eqnarray*}
 S^T_{j_1i_2} &= & \sum_s T^T_{j_1s}D(\alpha,i_1,i_2)_{si_2}\\
 & = & T^T_{j_1i_2}+ \alpha T^T_{j_1i_1}\\
 & = & T^T_{j_1i_2} - \frac{T^T_{j_1i_1} T_{i_2j_1}}{T_{i_1j_1}}\\
 & = & T^T_{j_1i_2} - \frac{T^T_{j_1i_1} T^T_{j_1i_2}}{T^T_{j_1i_1}}\\
 & = & 0
  \end{eqnarray*}
\section{Algorithms of the Attack}\label{sec:algo}

\begin{algorithm}[t!]
\caption{Algorithm for detecting the positions of degree $2$ \label{al:detection2}}
{\bf Function} {\tt IsADegree2Position}$(\CC, i, s_{\textrm{max}})$\\
{\bf requires:} 
\begin{itemize}
\item a code $\CC$ which is of the same form as $\Cpub^\perp$ 
of the BBCRS scheme;
\item $i$  a code position of $\CC$; 
\item a maximal number $s_{\textrm{max}}$ of tests.
\end{itemize}
{\bf Output:} {\tt yes} (if $i$ has degree $2$)/{\tt probably not} (if we think that $i$ has degree $1$).
\begin{algorithmic}
\FOR{$s=1$ {\bf to} $s_{\text{max}}$}
\STATE{$\Ind \leftarrow$ Random subset of $\{1,\dots,n\} \setminus \{i\}$
which satisfies \eqref{eq:dist_rule}}
\STATE{$\code{D} \leftarrow \sh{\Ind}{\CC}$}
\IF{{\tt Dimension}($\code{D}^2$) $\neq$ {\tt Dimension}($\pu{i}{\code{D}^2}$)}
\STATE{{\bf return} {\tt yes}}
\ENDIF
\ENDFOR
\STATE{{\bf return} {\tt  probably not}}
\end{algorithmic}
\end{algorithm}

\begin{algorithm}[t!]
\caption{Algorithm to compute quickly the sets $\Jind_1$ and $\Jind_2$ of positions of
degree $1$ and $2$ \label{al:compute_degree2}}
{\bf Function} {\tt Degree2andPositions}$(\CC, s_{\textrm{max}})$\\
{\bf requires:} 
\begin{itemize}
\item a code $\CC$ which is of the same form as $\Cpub^\perp$ 
of the BBCRS scheme;
\item a maximal number $s_{\textrm{max}}$ of tests.
\end{itemize}
{\bf Output:} The set $\Jind_2$ of positions of degree $2$.
\begin{algorithmic}
\STATE{$\Jind_1 \leftarrow \{1, \ldots, n\}$}
\STATE{$\Jind_2 \leftarrow \{\}$}
\FOR{$s=1$ {\bf to} $s_{\text{max}}$}
\STATE{$\Ind \leftarrow$ Random subset of $\{1,\dots,n\}$
which satisfies \eqref{eq:dist_rule}}
\STATE{$\code{D} \leftarrow \sh{\Ind}{\CC}$}
\FOR{$i \in \Jind_1 \setminus \Ind$}
\IF{{\tt Dimension}($\code{D}^2$) $\neq$ {\tt Dimension}($\pu{i}{\code{D}^2}$)}
\STATE{$\Jind_1 \leftarrow \Jind_1 \setminus \{i\}$}
\STATE{$\Jind_2 \leftarrow \Jind_2 \cup \{i\}$}
\ENDIF
\ENDFOR
\ENDFOR
\STATE{{\bf return} $\Jind_1, \Jind_2$}
\end{algorithmic}
\end{algorithm}

\begin{algorithm}[t!]
\caption{Algorithm to compute the set of positions of degree $2$ which are
associated to a given position $i_1$ of degree 1. That is positions $i_2$
such that $j(i_1) \in j(i_2)$ \label{al:Associated}}
{\bf Function} {\tt AssociatedDegree2Positions}$(\CC, i_1, \Jind_1, \Jind_2)$\\
{\bf requires:} 
\begin{itemize}
\item a code $\CC$ which is of the same form as $\Cpub^\perp$ 
of the BBCRS scheme;
\item The sets $\Jind_1, \Jind_2$ of positions of respective degrees
$1$ and $2$;
\item A position $i_1\in \Jind_1$;
\item a maximal number $s_{\textrm{max}}$ of tests.
\end{itemize}
{\bf Output:} The set of positions of degree $2$
associated to $i_1$.
\begin{algorithmic}
\STATE{$\CC_{i_1} \leftarrow \sh{i_1}{\CC}$}
\STATE{$\Jind_1, \Jind_2 \leftarrow $ {\tt Degree1and2Positions}$(\CC_{i_1}, s_{\textrm{max}})$}
\STATE{{\bf return} $\Jind_2 \setminus \Jind$}
\end{algorithmic}
\end{algorithm}

\begin{algorithm}[t!]
\caption{Algorithm to transform a degree $2$ position in a degree $1$
one \label{al:Eliminate}}
{\bf Function} {\tt EliminateDegree2Position}$(\CC, i_1, i_2, s_{\textrm{max}})$\\
{\bf requires:} 
\begin{itemize}
\item A code $\CC$ which is of the same form as $\Cpub^\perp$ 
of the BBCRS scheme; 
\item A position $i_2 \in \Jind_2$;
\item A position $i_1 \in \Jind_1$ associated to $i_2$;
\item A maximal number of tests $s_{\textrm{max}}$
\end{itemize}
{\bf Output:} A pair $(\CC', \alpha)$, where $\CC' = \CC \Dm(\alpha, i_1, i_2)$
and $\Dm (\alpha, i_1, i_2)$ is defined in Proposition \ref{pr:plus}.

\begin{algorithmic}
\FOR{$\alpha \in \F_q$}
\STATE{$\CC' \leftarrow \CC \Dm(\alpha, i_1, i_2)$}
\IF{{\tt IsADegree2Position}($\CC', i_2, s_{\textrm{max}}$) = {\tt false}}
\STATE{{\bf return} $\CC', \alpha$}
\ENDIF
\ENDFOR
\STATE{{\bf return} ``{\tt ERROR, $i_1$ is not associated to $i_2$}''}
\end{algorithmic}
\end{algorithm}

\begin{algorithm}[t!]
\caption{Complete algorithm for the attack \label{al:complete_attack}}
{\bf Function} {\tt CompleteAttack}$(\CC)$\\
{\bf requires:} 
\begin{itemize}
\item A code $\CC$ which is of the same form as $\Cpub^\perp$ 
of the BBCRS scheme. That is, $\CC = \GRS{n-k}{\xv}{\yv}(\Tm + \Rm)^T$
for some sparse matrix $\Tm$ and some rank one matrix $\Rm$.
The set of degree $1$ positions of $\CC$ contains an information set.
\item A maximal number of tests $s_{\textrm{max}}$
\end{itemize}

{\bf Output:} A tuple
$(\widetilde{\Tm}, \widetilde{\Rm}, \uv, \vv, \Ind)$ such that
$\pu{\Ind}{\CC} = \GRS{n-k}{\uv}{\vv}(\widetilde{\Tm} +\widetilde{\Rm})^T$
and with $\Ind$ as large as possible.
\begin{algorithmic}
\STATE{$\CC' \leftarrow \CC$}
\STATE{$\Jind_1, \Jind_2 \leftarrow $ {\tt Degree1and2Positions}($\CC, s_{\textrm{max}}$)}
\STATE{$\widetilde{\Tm} \leftarrow \Im$}
\COMMENT{$\Im$ is the $n\times n$ identity matrix.}
\WHILE{$\Jind_2 \neq \emptyset$ {\bf and} $\Jind_1 \neq \emptyset$ {\bf do}}
\STATE{$i_1 \leftarrow $ Random($\Jind_1$)}
\STATE{$\Jind' \leftarrow$ {\tt AssociatedDegree2Positions}($\CC', i_1, \Jind_1, \Jind_2$)}
\FOR{$i_2 \in \Jind'$ {\bf do}}
\STATE{$\CC'', \alpha \leftarrow $ {\tt EliminateDegree2Position}($\CC', i_1, i_2$)}
\STATE{$\CC' \leftarrow \CC''$}
\COMMENT{We replaced $\CC'$ by $\CC' \Dm (\alpha, i_1, i_2)$, where $\Dm(\alpha, i_1, i_2)$ is defined in Proposition \ref{pr:plus}}
\STATE{$\widetilde{\Tm}\leftarrow \Dm(\alpha, i_1, i_2)^{-1} \widetilde{\Tm}$}
\COMMENT{We preserve the loop invariant $\CC' = \CC \widetilde{T}$}
\STATE{$\Jind_2 \leftarrow \Jind_2 \setminus \{i_2\}$}
\STATE{$\Jind_1 \leftarrow \Jind_1 \setminus \{i_1\}$}
\STATE{$\Jind_1 \leftarrow \Jind_1 \cup \{i_2\}$}
\ENDFOR
\ENDWHILE
\STATE{$\code{C}' \leftarrow \pu{\Jind_2}{\code{C}'}$}
\STATE{$\widetilde{\Tm} \leftarrow $ {\tt Puncture}$(\widetilde{\Tm},\Jind_2)$}
\COMMENT{We drop the columns and the rows of $\widetilde{\Tm}$ that belong to $\Jind_2$.}\\
\COMMENT{At this point, there exist $\uv,\vv,\widetilde{\Pim},\widetilde{\Rm}$ with 
$\widetilde{\Rm}$ of rank one, $\widetilde{\Pim}$ being a permutation  matrix
such that $\pu{\Ind_2}{\code{C}} = \GRS{n-k}{\uv}{\vv}(\widetilde{\Pim}+\widetilde{\Rm})^T\widetilde{\Tm}$.}
\STATE{$(\widetilde{\Pim},\widetilde{\Rm},\uv,\vv) \leftarrow$ {\tt Attack}$(\code{C}')$}
\COMMENT{Here the algorithm of \cite[\S 4]{CGGOT:dcc14} is used and give a possible $\widetilde{\Pim}$,$\widetilde{\Rm}$, $\uv$ and $\vv$ that satisfy
$\code{C}' = \GRS{n-k}{\uv}{\vv}(\widetilde{\Pim}+\widetilde{\Rm})^T$. }
\STATE{{\bf return} $(\widetilde{\Tm}^T \widetilde{\Pim},\widetilde{\Tm}^T\widetilde{\Rm},\uv,\vv, \Jind_2)$}
\end{algorithmic}
\end{algorithm}


\end{document}